\documentclass[11pt]{article}

\usepackage[a4paper,margin=1in]{geometry}
\usepackage{amsmath,amssymb,amsthm,mathtools,bm,mathrsfs}
\usepackage{graphicx}
\usepackage{microtype}
\usepackage{hyperref}
\hypersetup{hidelinks}
\usepackage{enumitem}
\usepackage{booktabs}
\numberwithin{equation}{section}

\newcommand{\R}{\mathbb R}
\newcommand{\C}{\mathbb C}
\newcommand{\Sph}{\mathbb S}
\newcommand{\dd}{\,\mathrm d}
\newcommand{\Dom}{\operatorname{Dom}}
\newcommand{\spec}{\operatorname{spec}}
\newcommand{\specess}{\operatorname{spec}_{\mathrm{ess}}}

\newcommand{\e}{\mathrm e}

\newtheorem{theorem}{Theorem}[section]
\newtheorem{proposition}[theorem]{Proposition}
\newtheorem{lemma}[theorem]{Lemma}

\theoremstyle{definition}

\title{\textbf{Bargmann--Fock Representation and Global Estimates\\
for the Linearized Hard-Sphere Boltzmann Operator}}

\author{Ilya Karlin\thanks{ikarlin@ethz.ch}\\
Department of Mechanical and Process Engineering\\
ETH Zurich, CH-8092 Zurich, Switzerland}

\date{\today}

\begin{document}
\maketitle

\begin{abstract}
We consider the linearized Boltzmann collision operator
for three-dimensional hard spheres and ask a representation-theoretic
question: can one retain the complete collision operator while reorganizing it
so that its angular and radial structures become explicit before any Sonine
truncation is imposed?

Fock representation is used here as the canonical bosonic-oscillator
encoding of the Gaussian-weighted Hermite expansion.  It introduces no quantum
dynamics and no additional kinetic variables.  Its usefulness is algebraic:
velocity translations, rotations, and polynomial moments acquire simple
operator realizations.  Starting from the Carleman representation, the
hard-sphere collision integral becomes a superposition of two orthogonal
translations.  Passing to the Hermite/Fock and Bargmann realizations turns
those translations into exponentials and produces a single analytic
coherent-state kernel generating all Hermite collision brackets.  Rotational
covariance then separates the problem into fixed angular-momentum sectors,
and each fixed $(\ell,m)$ sector reduces exactly to a one-complex-variable
radial Bargmann space, organized by a $\mathfrak{su}(1,1)$ dynamical algebra indexed by radial excitation levels $j$.
The complete radial operator on that space is
\[
 \mathfrak A_\ell=\kappa\pi G(\mathcal J_\ell)-\mathfrak C_\ell,
\]
where $\mathcal J_\ell$ is an explicit Laguerre--Jacobi differential operator,
$G$ is the exact hard-sphere collision-frequency function, and
$\mathfrak C_\ell$ is compact.  An explicit unitary
intertwiner maps this analytic radial realization to the classical
Burnett/Sonine radial coordinate.  The Burnett matrix is therefore recovered
as a coordinate matrix of the exact radial operator rather than introduced as
its definition.  The low stress and heat-flux blocks and the classical Sonine
transport corrections are reproduced as normalization and consistency checks.

The main new results concern the untruncated radial operator.  The loss--gain
split becomes a spectral division of labor: the compact gain supplies compact
corrections, whereas the loss block $\kappa\pi G(J_\ell)$ is noncompact and
sets the common essential band $[\nu_{\min},\infty)$.  The Fock
representation makes access to that continuum constructive.  Radially
squeezed $\mathfrak{su}(1,1)$ coherent packets give a threshold Weyl sequence in every
fixed angular sector.  
If instead \(\ell\to\infty\) while the squeezing is tuned on the \(\ell^{-1}\) scale,
the same normalized packets concentrate at
any prescribed radial energy $x_0>0$ and furnish full-space Weyl sequences at
$\kappa\pi G(x_0)$ throughout the interior of the essential band.  Their expected
radial location obeys $\langle j\rangle\sim \ell^2$, exposing
quadratic continuum corridors in the $(\ell,j)$ plane and showing why no
fixed Sonine depth can resolve the threshold region.  At fixed radial depth
and $\ell\to\infty$ the loss is relatively diagonal, whereas at fixed $\ell$
and $j\to\infty$ the local matrix tail approaches a universal Toeplitz
operator whose zero at quasi-momentum $\theta=\pi$ is the local signature of
the same squeezed threshold escape.  Thus the payoff of Fock reformulation
for hard spheres is not exact degree grading, as for Maxwell molecules, but
exact radialization together with constructive, nonperturbative control of the
global spectral geometry.
\end{abstract}

\tableofcontents

\newpage

\section{Introduction}

The linearized hard-sphere Boltzmann collision operator is a classical object
with two complementary descriptions.  In the weighted velocity Hilbert space
one writes
\begin{equation}
 \mathcal A:=-L=\nu-\mathcal C,
 \label{eq:intro_loss_gain}
\end{equation}
where $\nu$ is the collision-frequency multiplication operator and
$\mathcal C$ is a compact integral operator.  This representation is the
natural starting point for spectral theory: the essential spectrum is
controlled by $\nu$, while isolated eigenvalues encode the nonlocal collision
dynamics.  A second classical description expands the operator in
Grad--Hermite or Burnett/Sonine polynomials.  Rotational symmetry then reduces
the problem to an infinite radial matrix for each angular momentum $\ell$.

The aim of the present paper is not to replace either description.  We ask
instead whether the complete operator can be placed in a representation in
which the connection between the classical loss--gain spectral description
and the rotationally reduced Hermite/Burnett radial description is explicit
\emph{before} any radial truncation is made.
The construction proceeds through
\begin{equation}
 \begin{aligned}
 \text{Carleman geometry}
 &\longrightarrow
 \text{Hermite/Fock representation}
 \longrightarrow
 \text{Bargmann generating kernel}
 \\
 &\longrightarrow
 SO(3)\text{ reduction}
 \longrightarrow
 \text{exact radial operator}.
 \end{aligned}
 \label{eq:intro_route}
\end{equation}
The final Burnett/Sonine matrices are then recovered as coordinates of that
operator rather than used to define it.

\subsection{What ``Fock'' means here, and relation to the earlier papers}

The terminology ``Fock space'' is used in a deliberately modest sense.  The
Gaussian-weighted Hermite expansion of a velocity-dependent function is
canonically isomorphic to the bosonic oscillator Fock space.  Hermite degree
is oscillator number, multiplication and differentiation in velocity become
creation and annihilation operations, and the Burnett basis is the
$SO(3)$-irreducible reorganization of the same Hermite states.  No quantum
dynamics or additional physical variables are introduced.  The advantage is
that elementary kinetic operations can be intertwined with simple algebraic
operators and then evaluated in whichever coordinate realization is most
economical.

This point of view was developed in two earlier works, but the structural
payoff was different in each case.  For the Lebowitz--Frisch--Helfand kinetic
model \cite{KarlinLFH2026}, the velocity-space relaxation operator becomes a
number-operator problem and the main issue is the exact intertwining of the
complete kinetic propagation when the realization depends on local
hydrodynamic fields.  That work established realization covariance and
showed how viscous stress and heat flux are selected by the graded Fock
structure in the hydrodynamic limit.

For Maxwell molecules \cite{KarlinMaxwell2026}, the object is instead the
nonlinear Boltzmann collision map.  Bobylev's substitution--product structure
leads to a representation-independent lift--fusion construction on symmetric
Fock space and to the exact degree rule
\begin{equation}
 \widehat{\mathcal Q}(\mathscr H_p,\mathscr H_q)
 \subseteq \mathscr H_{p+q}.
 \label{eq:intro_maxwell_grading}
\end{equation}
The exact grading explains the triangular moment hierarchy and makes the
stress and heat-flux relaxation rates genuine eigenvalues.

Hard spheres provide a different test of the same operator-first approach.
The relative-speed factor destroys the Maxwell degree grading, so one should
not expect finite Fock levels to remain invariant.  What survives is
intertwining, exact rotational reduction, and a canonical radial organization.
The main question of this paper is therefore not whether hard spheres can be
made to look like Maxwell molecules, but what exact structure remains after
number grading is lost.

\subsection{Notation conventions}

Cartesian vectors are written in bold, for example $\bm v,\bm x,\bm z$.
Geometric unit vectors used in the Carleman frame are denoted by
$\bm\omega,\bm\tau$, while Cartesian Hermite multi-indices are denoted by
bold Greek multi-indices such as $\boldsymbol{\nu}$.  This avoids using the
same bold symbol both for a direction in velocity space and for an oscillator
number state.  Repeated Cartesian component indices
$\alpha,\beta,\ldots$ are summed.  The standard coordinate unit vectors are
$\bm e_1,\bm e_2,\bm e_3$.

A hat, as in $\widehat{\mathcal A}$, $\widehat K_+$, or
$\hat a_\alpha$, denotes an operator on the abstract Fock space.  Unhatted $K_\pm,K_0$ denote their three-dimensional Bargmann coordinate
realizations; after restriction to a fixed angular sector the induced
one-variable operators are written $K_\pm^{(\ell)},K_0^{(\ell)}$.  The symbol $\mathcal A^{(\ell)}$ denotes
the matrix of a fixed-$\ell$ block in the normalized radial ket basis, and
$J_\ell$ denotes the corresponding Jacobi matrix of radial energy.  The
fraktur symbol $\mathfrak A_\ell$ denotes the one-variable radial Bargmann
operator; its Burnett realization is written $\mathfrak A_\ell^{\rm Burn}$.
The analogous conventions are used for the gain and loss blocks.  Orthogonal
projections are written with sans-serif letters $\mathsf P,\mathsf Q$;
ordinary $P_\ell$ is reserved for the Legendre polynomial of degree $\ell$.
We use $\operatorname{spec}$ and $\operatorname{spec}_{\rm ess}$ for the
spectrum and essential spectrum, respectively.  New special-function notation
is defined when it first appears.

\subsection{Classical context and pertinent previous work}

Only the part of the extensive literature directly relevant to the present
construction is reviewed here.  Carleman's analysis of the Boltzmann integral
equation established the geometric representation that underlies our starting
point \cite{Carleman1933}.  For cut-off hard interactions, Grad showed that the
linearized collision operator can be written as a collision-frequency
multiplication operator plus a compact integral operator; this identifies the
essential spectrum by compact-perturbation theory
\cite{GradRarefied1949,Grad1963,Drange1975}.  For rigid spheres, Pekeris also
established square-integrability properties of the linearized integral kernel
\cite{Pekeris1963}.

The classical Chapman--Enskog solution and its Sonine-polynomial evaluation of
transport coefficients are reviewed systematically by Chapman and Cowling
\cite{ChapmanCowling1970}.  The eigenvalue problem for a gas of rigid spheres
was studied computationally and analytically by Alterman, Frankowski, and
Pekeris \cite{AltermanFrankowskiPekeris1962}.  Their work, together with the
related sound-propagation calculations \cite{PekerisEtAl1962}, demonstrated
the practical effectiveness of radial polynomial expansions.  Kumar clarified
the relation among Grad, Burnett/Sonine, and irreducible polynomial systems
and emphasized the economy of the Burnett organization \cite{Kumar1966}.
Ford derived general Burnett matrix elements for power-law interactions, with
elastic spheres as a special case \cite{Ford1968}.  A very recent spectral
implementation by Hiemstra, Ke\ss ler, and Abdelmalik uses Wigner--Eckart
factorization to separate universal $SO(3)$ angular couplings from reduced
radial collision data and thereby to construct finite spectral collision
tensors efficiently \cite{HiemstraEtAl2026}.  This is complementary to the
question pursued here: after exploiting the same rotational structure, we
retain the radial problem as an exact infinite-dimensional operator rather
than selecting a finite Burnett/Sonine block.

A spectral feature particularly relevant below is the angular decomposition.
Klaus studied the hard-sphere discrete spectrum in the invariant angular
subspaces and proved disappearance of discrete eigenvalues beyond a critical
angular momentum; his computer-assisted analysis indicated that $\ell=3$ is
the first angular sector without discrete eigenvalues \cite{Klaus1976}.
Constructive spectral-gap and coercivity estimates for hard potentials,
including hard spheres, were developed by Baranger and Mouhot and by Mouhot
\cite{BarangerMouhot2005,Mouhot2006}.  Dudy\'nski treats the velocity-space
collision operator separately from the Fourier-transformed and full
transport--collision generators; for the collision operator his Theorem~3.3
excludes threshold accumulation and eigenvalues of infinite multiplicity
under an integrability condition satisfied by hard spheres
\cite{Dudynski2013}.  This distinction is relevant here because the present
paper concerns the spatially homogeneous, self-adjoint collision operator
only.

Bobylev and Mossberg returned specifically to the three-dimensional
hard-sphere operator and reduced the radial integral equations to systems of
ordinary differential equations.  They used this reduction both for
eigenvalue computations and for large-energy asymptotics
\cite{BobylevMossberg2008}.  Their work is particularly close in spirit to the
present one because both approaches seek an exact radial formulation rather
than a finite Sonine truncation.  Here, however, the radial equation is
obtained by first constructing and then intertwining the complete
Carleman--Bargmann operator.

\subsection{Road map and main results}

The paper is intended to be read as a sequence of reductions rather than as a
collection of independent formulas.  Section~\ref{sec:carleman_bargmann} starts from the Carleman
representation and explains why the Hermite/Fock and Bargmann realizations are
adapted to its translation structure.  The output is a single analytic kernel
$\mathcal K(\bm w,\bm z)$ whose Taylor coefficients generate all Cartesian
Hermite collision brackets.

Section~\ref{sec:angular_radial} uses rotational invariance.  The labels have a direct kinetic
meaning: $\ell$ is angular tensor rank, $m$ labels the $2\ell+1$ equivalent
orientations, and $j=0,1,\ldots$ counts radial excitations within a fixed
angular sector.  The radial index is constructed, rather than imposed: a
rotationally scalar pair creator built from the canonical Fock ladder
operators generates the complete radial \emph{tower} above each harmonic
lowest-weight state.  The three-dimensional coherent kernel is then projected
onto these towers, producing an exact one-complex-variable radial operator
kernel $\mathscr A_\ell$ for every $\ell$.

Section~\ref{sec:jacobi_intertwining} isolates the part of the radial operator that is already known
exactly from the collision frequency.  All dependence on the normalization
of the hard-sphere cross section is kept in one constant $\kappa>0$.  The
result is
\begin{equation}
 \boxed{
 \mathfrak A_\ell
 =
 \kappa\pi\,G(\mathcal J_\ell)-\mathfrak C_\ell,
 }
 \label{eq:intro_main_operator}
\end{equation}
where $\mathfrak C_\ell$ is compact and
\begin{equation}
 G(x)
 =
 \sqrt{\frac2\pi}\e^{-x}
 +
 \left(\sqrt{2x}+\frac1{\sqrt{2x}}\right)\operatorname{erf}(\sqrt x),
 \label{eq:intro_G}
\end{equation}
where $\operatorname{erf}$ is the standard error function.  The radial
parameter $\beta_\ell$ and the operator $\mathcal J_\ell$ are
\begin{equation}
 \mathcal J_\ell
 =
 \xi\partial_\xi^2
 +(2\xi+\beta_\ell)\partial_\xi
 +\xi+\beta_\ell,
 \qquad
 \beta_\ell=\ell+\frac32.
 \label{eq:intro_J}
\end{equation}
Here $\xi$ is the one-complex-variable radial Bargmann coordinate: after a
solid harmonic of degree $\ell$ is factored from the three-dimensional Bargmann
polynomial, $\xi=\bm z^2/2$ records the remaining scalar pair excitations.
An exact unitary transform then identifies this analytic radial realization
with the conventional Laguerre/Burnett radial coordinate.  Only after this
intertwining theorem do we return to Burnett/Sonine language: the conventional
collision matrices are recovered as coordinates of the already constructed
radial operator, and their lowest hard-sphere blocks are derived explicitly as
normalization and consistency checks.

Section~\ref{sec:burnett_validation} then returns to the familiar low-order
sectors.  The stress block in the $\ell=2$ tower, the heat-flux block in the
$\ell=1$ tower, and finite Sonine compressions are extracted from the same
exact radial kernel.  Their agreement with the classical hard-sphere
transport sequence is used only as validation: no finite radial truncation
enters the construction of the operator itself.

Having fixed the normalization in this classical regime, the remainder of the
paper leaves finite Sonine windows and turns to global questions in Sections \ref{sec:global_spectral}, \ref{sec:large_l}, \ref{sec:large_j} and \ref{sec:semiclassical_synthesis}.  The guiding
observation is that the Boltzmann loss--gain split becomes a spectral division
of labor: the loss determines the noncompact radial geometry, while the
compact gain produces compact corrections and possible discrete levels on that
background.  The classical hard-sphere collision frequency has the minimum
\begin{equation}
 \nu_{\min}=2\kappa\sqrt{2\pi},
 \label{eq:intro_numin}
\end{equation}
at zero velocity.  The common essential band follows by the standard invariance
of essential spectrum under compact perturbations \cite{ReedSimon1980},
but the radial representation goes further by displaying how its points are
reached in Fock/Burnett coordinates.  The squeezed radial family can be tuned
to concentrate at any prescribed physical radial energy $x_0>0$, producing a
Fock-space Weyl sequence at $\kappa\pi G(x_0)$; the threshold is the endpoint
$x_0\downarrow0$.  At fixed radial
depth and $\ell\to\infty$ the loss collapses, relative to its leading scale,
toward a diagonal operator and every fixed Sonine section rises as
$\sqrt\ell$.  An explicit normalized radial family nevertheless reaches the
threshold by escaping to large $j$; a fixed low-energy layer requires
$j\sim O(\ell^2)$, while convergence all the way to $x=0$ requires still faster
radial escape.  At fixed $\ell$ and $j\to\infty$ the local matrix tail instead
approaches a universal Toeplitz operator with symbol
$2\sqrt2\,|\cos(\theta/2)|$.  Here $\theta$ is not a physical velocity-space
angle but the Fourier phase (quasi-momentum) conjugate to translation in the
locally homogeneous radial-index lattice.  A threshold Weyl sequence has the
alternating phase $(-1)^j$, identifying the zero at $\theta=\pi$ as the
cancellation mechanism for the leading $\sqrt j$ loss scale.  The subsequent
semiclassical synthesis organizes the interior of the $(\ell,j)$ plane by the
total energy scale $E_{j\ell}=2j+\beta_\ell$ and the angular/radial balance
$\sigma_{j\ell}=\beta_\ell/E_{j\ell}$. Conclusions and discussion are provided in Sections \ref{sec:practical_regimes} and \ref{sec:discussion}.

\section{From Carleman translations to the Bargmann--Fock kernel}
\label{sec:carleman_bargmann}

\subsection{Carleman geometry and linearization}

We use the normalized equilibrium Maxwellian
\begin{equation}
 M(\bm v)=(2\pi)^{-3/2}\exp(-|\bm v|^2/2),
 \label{eq:Maxwellian_main}
\end{equation}
so velocities are dimensionless thermal velocities.  Let $Q(f,f)$ denote the
usual bilinear hard-sphere Boltzmann collision operator.  For three-dimensional
hard spheres its Carleman representation can be written, up to the positive
cross-section normalization constant $\kappa$, as
\begin{equation}
 Q(f,f)(\bm v)
 =
 \kappa
 \iint_{\R^3\times\R^3}
 \delta(\bm x\cdot\bm y)
 \left[
 f(\bm v+\bm x)f(\bm v+\bm y)
 -
 f(\bm v+\bm x+\bm y)f(\bm v)
 \right]
 \dd\bm x\,\dd\bm y .
 \label{eq:carleman}
\end{equation}
Here $\delta$ is the Dirac delta distribution; its argument enforces the
Carleman orthogonality condition $\bm x\cdot\bm y=0$.  The special
simplification for hard spheres in $d=3$ is that the transformed Carleman
kernel is constant \cite{MouhotPareschi2006}.

Linearize about $M$ by writing
\begin{equation}
 f=M(1+\varepsilon\phi),
\end{equation}
where $\varepsilon$ is a formal perturbation parameter and $\phi$ is the
dimensionless perturbation.
On the support $\bm x\cdot\bm y=0$,
\begin{equation}
 M(\bm v+\bm x)M(\bm v+\bm y)
 =
 M(\bm v)M(\bm v+\bm x+\bm y),
\end{equation}
and the linearized operator is
\begin{equation}
 (L\phi)(\bm v)
 =
 \kappa
 \iint
 \delta(\bm x\cdot\bm y)\,
 M(\bm v+\bm x+\bm y)\,
 \Delta_{\bm x,\bm y}\phi(\bm v)
 \dd\bm x\,\dd\bm y ,
 \label{eq:linearized}
\end{equation}
with
\begin{equation}
 \Delta_{\bm x,\bm y}\phi
 =
 \phi(\bm v+\bm x)+\phi(\bm v+\bm y)
 -\phi(\bm v+\bm x+\bm y)-\phi(\bm v).
 \label{eq:collision_difference}
\end{equation}
Equivalently,
\begin{equation}
 \Delta_{\bm x,\bm y}
 =-(T_{\bm x}-I)(T_{\bm y}-I),
 \label{eq:delta_translation}
\end{equation}
where $T_{\bm x}$ is the velocity-translation operator,
$(T_{\bm x}\phi)(\bm v)=\phi(\bm v+\bm x)$, and $I$ is the identity.

\subsection{Minimal Hermite/Fock dictionary}
The previous subsection has rewritten a collision as a superposition of
velocity translations.  Before exploiting that fact, we state explicitly
what the Fock notation means in this paper.

Standard references for the oscillator/Fock and Bargmann realizations are
Bargmann's original construction and the modern account of Folland
\cite{Bargmann1961,Folland1989}; we use the same abstract-versus-realization
convention as in the Maxwell-molecule formulation \cite{KarlinMaxwell2026}.

Let $\mathbb N_0=\{0,1,2,\ldots\}$ and let
$\boldsymbol{\nu}=(\nu_1,\nu_2,\nu_3)\in\mathbb N_0^3$ be a Cartesian
multi-index, with $|\boldsymbol{\nu}|=\nu_1+\nu_2+\nu_3$ and
$\boldsymbol{\nu}!=\nu_1!\nu_2!\nu_3!$.  The vacuum is
$|0\rangle:=|0,0,0\rangle$.  Creation and annihilation operators are denoted
by $\hat a_\alpha^\dagger$ and $\hat a_\alpha$ and satisfy
\begin{equation}
 [\hat a_\alpha,\hat a_\beta^\dagger]=\delta_{\alpha\beta},
 \qquad
 \hat a_\alpha|0\rangle=0,
\end{equation}
where $[A,B]=AB-BA$ and $\delta_{\alpha\beta}$ is the Kronecker delta.  The
normalized number states are
\begin{equation}
 |\boldsymbol{\nu}\rangle
 =
 \prod_{\alpha=1}^3
 \frac{(\hat a_\alpha^\dagger)^{\nu_\alpha}}{\sqrt{\nu_\alpha!}}
 |0\rangle .
 \label{eq:number_state_dictionary}
\end{equation}
The symmetric Fock space is the Hilbert-space organization
\begin{equation}
 \mathscr H
 =
 \bigoplus_{N=0}^{\infty}\mathscr H_N,
 \qquad
 \mathscr H_N
 =
 \operatorname{span}\{
 |\boldsymbol{\nu}\rangle:\ |\boldsymbol{\nu}|=N
 \}.
 \label{eq:fock_grading_dictionary}
\end{equation}
The integer $N$ is simply total Hermite degree.  The number operator
\begin{equation}
 \widehat N=\hat a_\alpha^\dagger\hat a_\alpha
\end{equation}
has eigenvalue $|\boldsymbol{\nu}|$ on $|\boldsymbol{\nu}\rangle$.

We shall use two exact coordinate realizations of the same abstract state.
The Hermite-polynomial realization $\mathcal R_P$ is defined on number states
by
\begin{equation}
 \mathcal R_P|\boldsymbol{\nu}\rangle
 =
 \frac{\operatorname{He}_{\boldsymbol{\nu}}(\bm v)}
 {\sqrt{\boldsymbol{\nu}!}},
 \label{eq:RP_dictionary}
\end{equation}
where $\operatorname{He}_{\boldsymbol{\nu}}$ is the product of probabilists'
Hermite polynomials associated with the Maxwellian in
Eq.~\eqref{eq:Maxwellian_main}.  The Bargmann realization $\mathcal R_B$ is
\begin{equation}
 \mathcal R_B|\boldsymbol{\nu}\rangle
 =
 \frac{\bm z^{\boldsymbol{\nu}}}{\sqrt{\boldsymbol{\nu}!}},
 \qquad
 \bm z=(z_1,z_2,z_3)\in\C^3.
 \label{eq:RB_dictionary}
\end{equation}
The maps $\mathcal R_P$ and $\mathcal R_B$ are invertible realizations;
changing between them changes coordinates, not the underlying state or
operator.  For any operator $O$ in a realization $\mathcal R$, its abstract
Fock representative is defined by the intertwining relation
\begin{equation}
 \widehat O=\mathcal R^{-1}O\mathcal R,
 \qquad\hbox{equivalently}\qquad
 O\mathcal R=\mathcal R\widehat O.
 \label{eq:generic_intertwining_dictionary}
\end{equation}
The standard ladder dictionaries are
\begin{align}
 \mathcal R_P\hat a_\alpha\mathcal R_P^{-1}
 &=\partial_{v_\alpha},
 &
 \mathcal R_P\hat a_\alpha^\dagger\mathcal R_P^{-1}
 &=v_\alpha-\partial_{v_\alpha},
 \label{eq:Hermite-ladders-nl}
 \\
 \mathcal R_B\hat a_\alpha\mathcal R_B^{-1}
 &=\partial_{z_\alpha},
 &
 \mathcal R_B\hat a_\alpha^\dagger\mathcal R_B^{-1}
 &=z_\alpha.
 \label{eq:Bargmann-ladders-nl}
\end{align}
It is useful to name also the two direct changes of realization,
\[
 \mathcal T_{P\to B}:=\mathcal R_B\mathcal R_P^{-1},
 \qquad
 \mathcal T_{B\to P}:=\mathcal R_P\mathcal R_B^{-1}
 =\mathcal T_{P\to B}^{-1}.
\]
Thus the Hermite and Bargmann descriptions are not two different kinetic
problems: they are two coordinate realizations of the same abstract Fock
state and are connected directly by $\mathcal T_{P\to B}$ and
$\mathcal T_{B\to P}$.  Schematically,
\begin{equation}
 \begin{array}{ccc}
 &\mathscr H&\\[-1mm]
 {\scriptstyle\mathcal R_P}\swarrow
 &&\searrow{\scriptstyle\mathcal R_B}\\[-1mm]
 \text{Hermite velocity functions}
 &
 \underset{\scriptstyle\mathcal T_{B\to P}}
 {\overset{\scriptstyle\mathcal T_{P\to B}}{\rightleftarrows}}
 &
 \text{Bargmann analytic functions}.
 \end{array}
 \label{eq:realization_diagram}
\end{equation}
Later, after fixing angular momentum, a third exact map
$\mathcal U_\ell$ will connect the radial Bargmann function to the physical
Laguerre/Burnett radial coordinate.
A hat will always indicate an operator acting on the abstract Fock space.  For
Maxwell molecules the eigenspaces of $\widehat N$ are preserved by the
linearized collision operator; one purpose of the hard-sphere calculation is
to see precisely how that property fails.
Throughout the geometric Carleman calculation, bold Greek letters such as
$\bm\omega$ and $\bm\tau$ denote ordinary unit vectors in velocity space.
They are deliberately kept distinct from the multi-index
$\boldsymbol{\nu}$ and from number states.

\subsection{From Carleman translations to the coherent Bargmann kernel}
\label{sec:coherent_bargmann_kernel}

We now apply the realization diagram \eqref{eq:realization_diagram} to the
linearized collision operator.  It is convenient to work with the positive
operator
\begin{equation}
 \mathcal A:=-L
\end{equation}
in the Hermite velocity realization and to denote its exact abstract Fock
representative and Bargmann realization by
\begin{equation}
 \widehat{\mathcal A}_{\rm HS}
 :=
 \mathcal R_P^{-1}\mathcal A\mathcal R_P,
 \qquad
 \mathcal A_B
 :=
 \mathcal R_B\widehat{\mathcal A}_{\rm HS}\mathcal R_B^{-1}.
 \label{eq:A_three_realizations}
\end{equation}
Equivalently,
\begin{equation}
 \mathcal A_B
 =
 \mathcal T_{P\to B}\,\mathcal A\,\mathcal T_{B\to P}.
 \label{eq:A_PB_intertwining}
\end{equation}
Thus every analytic object constructed below belongs to the Bargmann
realization of the same linearized hard-sphere operator; no auxiliary kinetic 
operator is being introduced.

The usefulness of the Fock representation follows directly from the
translation structure in Eq.~\eqref{eq:delta_translation}.  In the
Hermite-polynomial realization,
\begin{equation}
 \mathcal R_P^{-1}T_{\bm x}\mathcal R_P
 =
 e^{\bm x\cdot\hat{\bm a}},
 \label{eq:translation_Fock}
\end{equation}
and therefore
\begin{equation}
 \widehat\Delta_{\bm x,\bm y}
 =
 -
 \left(\e^{\bm x\cdot\hat{\bm a}}-\hat I\right)
 \left(\e^{\bm y\cdot\hat{\bm a}}-\hat I\right).
 \label{eq:Delta_Fock}
\end{equation}
If the passive Bargmann coordinate is denoted temporarily by $\bm\zeta$, then
$\mathcal R_B\hat a_\alpha\mathcal R_B^{-1}=\partial_{\zeta_\alpha}$, so
$\exp(\bm x\cdot\hat{\bm a})$ is represented by translation of the Bargmann
argument.  The scalar exponentials that will enter the collision kernel arise
more specifically because coherent vectors diagonalize all annihilation
operators.

Indeed, introduce the unnormalized coherent vector
\cite{Bargmann1961,Folland1989}
\begin{equation}
 |\bm z\rangle_{\rm coh}
 =
 \e^{\bm z\cdot\hat{\bm a}^\dagger}|0\rangle,
 \qquad
 \hat{\bm a}|\bm z\rangle_{\rm coh}
 =\bm z|\bm z\rangle_{\rm coh}.
 \label{eq:coherent_vector}
\end{equation}
The same abstract coherent vector then has the two coordinate images
\begin{align}
 (\mathcal R_B|\bm z\rangle_{\rm coh})(\bm\zeta)
 &=
 \e^{\bm z\cdot\bm\zeta},
 \label{eq:coherent_Bargmann_image}
 \\
 (\mathcal R_P|\bm z\rangle_{\rm coh})(\bm v)
 &=
 E_{\bm z}(\bm v)
 :=
 \exp\!\left(\bm z\cdot\bm v-\frac{\bm z^2}{2}\right).
 \label{eq:Ez}
\end{align}
The second identity is precisely the multivariate Hermite generating
function, obtained by expanding the coherent vector in number states and
using Eq.~\eqref{eq:RP_dictionary}.  The first identity is its Bargmann
counterpart.  In particular,
\begin{equation}
 \e^{\bm x\cdot\hat{\bm a}}|\bm z\rangle_{\rm coh}
 =
 \e^{\bm x\cdot\bm z}|\bm z\rangle_{\rm coh},
 \label{eq:coherent_translation_eigen}
\end{equation}
which in the Hermite realization becomes
\begin{equation}
 E_{\bm z}(\bm v+\bm x)
 =
 E_{\bm z}(\bm v)\e^{\bm z\cdot\bm x}.
 \label{eq:Ez_translation}
\end{equation}
Consequently the Carleman collision difference factorizes on a coherent
vector:
\begin{equation}
 \Delta_{\bm x,\bm y}E_{\bm z}
 =
 -
 E_{\bm z}(\bm v)
 \left(\e^{\bm z\cdot\bm x}-1\right)
 \left(\e^{\bm z\cdot\bm y}-1\right).
 \label{eq:coherent_factorization}
\end{equation}

This suggests evaluating the exact Fock operator on coherent vectors.  Define
its polarized coherent-state, or Bargmann, kernel by
\begin{equation}
 \mathcal K(\bm w,\bm z)
 :=
 {}_{\rm coh}\langle\bm w|
 \widehat{\mathcal A}_{\rm HS}
 |\bm z\rangle_{\rm coh}.
 \label{eq:Kdef}
\end{equation}
Here $\bm w$ is an independent analytic dual variable.  Expanding the bra and
ket in the normalized number basis gives immediately
\begin{equation}
 \mathcal K(\bm w,\bm z)
 =
 \sum_{\boldsymbol{\mu},\boldsymbol{\nu}}
 \mathcal A_{\boldsymbol{\mu}\boldsymbol{\nu}}
 \frac{\bm w^{\boldsymbol{\mu}}\bm z^{\boldsymbol{\nu}}}
 {\sqrt{\boldsymbol{\mu}!\boldsymbol{\nu}!}},
 \qquad
 \mathcal A_{\boldsymbol{\mu}\boldsymbol{\nu}}
 =
 \langle\boldsymbol{\mu}|
 \widehat{\mathcal A}_{\rm HS}
 |\boldsymbol{\nu}\rangle.
 \label{eq:Kgenerates}
\end{equation}
Thus $\mathcal K$ is not merely a mnemonic generating function: it is the
polarized analytic matrix kernel of the exact Bargmann operator
$\mathcal A_B$.  Its Taylor coefficients are all Cartesian Hermite/Fock
matrix elements, and hence it determines $\mathcal A_B$ on the algebraic
Fock core.  In the standard Bargmann integral notation the second analytic
variable is replaced by the conjugate integration variable; this is the
usual reproducing-kernel realization of the same operator
\cite{Bargmann1961,Folland1989}.

\paragraph{Evaluation of the kernel in velocity coordinates.}
The Bargmann description explains what $\mathcal K$ is; the Hermite velocity
realization is the economical coordinate system in which to evaluate it,
because the Carleman formula is already written in $\bm v$.
For an actual Hilbert-space matrix element the first coherent label is
complex-conjugated.  We compute that matrix element and then use analytic
polarization, writing the conjugate label as the independent variable
$\bm w$.  With
\begin{equation}
 \langle\psi,\phi\rangle_M
 :=
 \int_{\R^3}
 M(\bm v)\overline{\psi(\bm v)}\phi(\bm v)\dd\bm v,
\end{equation}
the symmetric Dirichlet form of the positive linearized operator is
\begin{equation}
 \langle\psi,\mathcal A\phi\rangle_M
 =\frac{\kappa}{4}
 \iiint
 \delta(\bm x\cdot\bm y)M(\bm v)M(\bm v+\bm x+\bm y)
 \overline{\Delta_{\bm x,\bm y}\psi(\bm v)}
 \Delta_{\bm x,\bm y}\phi(\bm v)
 \dd\bm v\,\dd\bm x\,\dd\bm y.
 \label{eq:carleman_dirichlet_form}
\end{equation}
This is the usual nonnegative quadratic form of the linearized Boltzmann
operator, written in the present Carleman variables; it follows by applying
the collision changes of variables to the four terms in
Eq.~\eqref{eq:collision_difference} (see, e.g.,
\cite{Grad1963,Drange1975}).

For compactness of notation set
\begin{equation}
 B_{\bm q}(\bm w,\bm z)
 :=
 \left(\e^{\bm w\cdot\bm q}-1\right)
 \left(\e^{\bm z\cdot\bm q}-1\right).
 \label{eq:Bq_def}
\end{equation}
Substituting the two coherent vectors and using
Eq.~\eqref{eq:coherent_factorization} on both legs gives the complete
Carleman integral for the generating kernel,
\begin{align}
 \mathcal K(\bm w,\bm z)
 =\frac{\kappa}{4}
 \iiint &\delta(\bm x\cdot\bm y)
 M(\bm v)M(\bm v+\bm x+\bm y)
 E_{\bm w}(\bm v)E_{\bm z}(\bm v)
 \notag\\[-1mm]
 &\times
 B_{\bm x}(\bm w,\bm z)
 B_{\bm y}(\bm w,\bm z)
 \dd\bm v\,\dd\bm x\,\dd\bm y.
 \label{eq:K_master_integral}
\end{align}
This is the point at which the abstract coherent kernel has been converted
back into an ordinary velocity-space integral.  The remaining reductions are
elementary, but we keep them explicit because they fix the normalization and
the analytic structure used throughout the rest of the paper.

Put $\bm s=\bm x+\bm y$.  The part depending on $\bm v$ is Gaussian:
\begin{align}
 &M(\bm v)M(\bm v+\bm s)
 E_{\bm w}(\bm v)E_{\bm z}(\bm v)
 \notag\\
 &\quad=
 \frac{1}{(2\pi)^3}
 \exp\!\left[
 -\left(
 \bm v-\frac{\bm w+\bm z-\bm s}{2}
 \right)^2
 -\frac{|\bm s|^2}{4}
 -\frac{(\bm z+\bm w)\cdot\bm s}{2}
 -\frac{(\bm z-\bm w)^2}{4}
 \right].
 \label{eq:v_gaussian_completed_square}
\end{align}
Therefore
\begin{align}
&\int_{\R^3}
 M(\bm v)M(\bm v+\bm s)
 E_{\bm w}(\bm v)E_{\bm z}(\bm v)\dd\bm v
=
\frac{1}{8\pi^{3/2}}
\exp\!\left[
-\frac{|\bm s|^2}{4}
-\frac{(\bm z+\bm w)\cdot\bm s}{2}
-\frac{(\bm z-\bm w)^2}{4}
\right].
\label{eq:v_gaussian_main}
\end{align}
Inserting this into Eq.~\eqref{eq:K_master_integral} gives the two-vector
Carleman integral
\begin{align}
 \mathcal K(\bm w,\bm z)
 ={}&
 \frac{\kappa}{32\pi^{3/2}}
 \exp\!\left[-\frac{(\bm z-\bm w)^2}{4}\right]
 \iint
 \delta(\bm x\cdot\bm y)
 \notag\\
 &\times
 \exp\!\left[
 -\frac{|\bm x+\bm y|^2}{4}
 -\frac{(\bm z+\bm w)\cdot(\bm x+\bm y)}{2}
 \right]
 B_{\bm x}(\bm w,\bm z)
 B_{\bm y}(\bm w,\bm z)
 \dd\bm x\,\dd\bm y.
 \label{eq:K_after_v}
\end{align}

We next resolve the orthogonality constraint by writing
\begin{equation}
 \bm x=r\bm\omega,\qquad
 \bm y=\rho\bm\tau,\qquad
 \bm\omega\cdot\bm\tau=0,
 \qquad r,\rho\ge0.
 \label{eq:Carleman_frame_param}
\end{equation}
If $\varphi$ is the angular coordinate on the great circle
$\bm\tau\perp\bm\omega$, then
\begin{equation}
 \delta(\bm x\cdot\bm y)\dd\bm x\,\dd\bm y
 =r\rho\,\dd r\,\dd\rho\,\dd\bm\omega\,\dd\varphi.
 \label{eq:carleman_radial_measure}
\end{equation}
Indeed, the ordinary measures contribute $r^2\rho^2$, while
$\delta(r\rho\,\bm\omega\cdot\bm\tau)$ removes one factor $r\rho$
and restricts $\bm\tau$ to the perpendicular great circle.  Since
$|\bm x+\bm y|^2=r^2+\rho^2$, both the Gaussian and the coherent factors in
Eq.~\eqref{eq:K_after_v} separate into an $r,\bm\omega$ part and a
$\rho,\bm\tau$ part.

For a frame direction $\bm\omega$ define
\begin{align}
 R_{\bm\omega}(\bm w,\bm z)
 :=\int_0^\infty &r\,
 \exp\!\left[-\frac{r^2}{4}
 -\frac{(\bm z+\bm w)\cdot\bm\omega}{2}r\right]
 \left(\e^{(\bm z\cdot\bm\omega)r}-1\right)
 \left(\e^{(\bm w\cdot\bm\omega)r}-1\right)\dd r.
 \label{eq:R_integral_main}
\end{align}
Then Eq.~\eqref{eq:K_after_v} becomes
\begin{align}
 \mathcal K(\bm w,\bm z)
 ={}&
 \frac{\kappa}{32\pi^{3/2}}
 \exp\!\left[-\frac{(\bm z-\bm w)^2}{4}\right]
 \int_{\Sph^2}\dd\bm\omega
 \int_{\bm\tau\perp\bm\omega}\dd\varphi\,
 R_{\bm\omega}(\bm w,\bm z)
 R_{\bm\tau}(\bm w,\bm z).
 \label{eq:K_frame_before_average}
\end{align}

The radial integral can be evaluated explicitly.  Set
\begin{equation}
 A_{\bm\omega}:=(\bm z+\bm w)\cdot\bm\omega,
 \qquad
 D_{\bm\omega}:=(\bm z-\bm w)\cdot\bm\omega.
\end{equation}
Expanding the two brackets in Eq.~\eqref{eq:R_integral_main} gives
\begin{align}
 R_{\bm\omega}
 =\int_0^\infty r\e^{-r^2/4}
 \big[&\e^{A_{\bm\omega}r/2}+\e^{-A_{\bm\omega}r/2}
 -\e^{D_{\bm\omega}r/2}-\e^{-D_{\bm\omega}r/2}\big]\dd r.
 \label{eq:R_four_exponentials}
\end{align}
Completing the square in each term yields
\begin{equation}
 R_{\bm\omega}(\bm w,\bm z)
 =2\sqrt\pi\,[F(A_{\bm\omega})-F(D_{\bm\omega})],
 \label{eq:Rclosed}
\end{equation}
where
\begin{equation}
 F(a)=a\,\e^{a^2/4}\operatorname{erf}(a/2).
 \label{eq:Fdef}
\end{equation}

The remaining angular measure is the invariant measure on ordered orthonormal
two-frames.  We normalize it by
\begin{equation}
 \int_{\Sph^2}\dd\bm\omega
 \int_{\bm\tau\perp\bm\omega}\dd\varphi\,(\cdot)
 =8\pi^2\langle\cdot\rangle_{\rm fr}.
 \label{eq:frame_average_definition}
\end{equation}
Substitution into Eq.~\eqref{eq:K_frame_before_average} finally gives
\begin{equation}
 \boxed{
 \mathcal K(\bm w,\bm z)
 =
 \frac{\kappa\sqrt\pi}{4}
 \exp\!\left[-\frac{(\bm z-\bm w)^2}{4}\right]
 \left\langle
 R_{\bm\omega}(\bm w,\bm z)
 R_{\bm\tau}(\bm w,\bm z)
 \right\rangle_{\rm fr}.
 }
 \label{eq:Kframe}
\end{equation}
Thus the prefactor $\kappa\sqrt\pi/4$ follows transparently from the
Dirichlet-form factor $\kappa/4$, the velocity Gaussian
$1/(8\pi^{3/2})$, and the frame volume $8\pi^2$.

For comparison with conventional spherical coordinates, write
\begin{equation}
 \bm\omega(\theta,\phi)
 =(
 \sin\theta\cos\phi,
 \sin\theta\sin\phi,
 \cos\theta),
\end{equation}
with $0\le\theta\le\pi$ and $0\le\phi<2\pi$.  Let
\begin{equation}
 \bm e_\theta
 =(
 \cos\theta\cos\phi,
 \cos\theta\sin\phi,
 -\sin\theta),
 \qquad
 \bm e_\phi=(-\sin\phi,\cos\phi,0),
\end{equation}
and parameterize the unit circle perpendicular to $\bm\omega$ by
$\bm\tau=\cos\psi\,\bm e_\theta+\sin\psi\,\bm e_\phi$,
$0\le\psi<2\pi$.  Then the frame average is
\begin{equation}
 \langle \Phi(\bm\omega,\bm\tau)\rangle_{\rm fr}
 =\frac{1}{8\pi^2}
 \int_0^{2\pi}\!\dd\phi
 \int_0^\pi\!\sin\theta\,\dd\theta
 \int_0^{2\pi}\!\dd\psi\,
 \Phi(\bm\omega(\theta,\phi),\bm\tau(\theta,\phi,\psi)).
 \label{eq:frame_average_spherical}
\end{equation}
This makes explicit that $\langle\cdot\rangle_{\rm fr}$ is simply the
normalized invariant average over ordered orthonormal two-frames.

For later coefficient extraction it is also useful to record the convergent
series
\begin{equation}
 R_{\bm\omega}
 =
 4\sum_{k=1}^{\infty}
 \frac{k!}{(2k)!}
 \left[
 A_{\bm\omega}^{2k}-D_{\bm\omega}^{2k}
 \right].
 \label{eq:Rseries}
\end{equation}
The non-polynomial entire function $F$ is the analytic signature of the loss
of Maxwell-molecule number grading.  Rotational covariance is now immediate:
\begin{equation}
 \mathcal K(R\bm w,R\bm z)=\mathcal K(\bm w,\bm z),
 \qquad R\in SO(3),
 \label{eq:rotcov}
\end{equation}
because the frame measure is invariant and every occurrence of the coherent
variables is through Euclidean scalar products.

\paragraph{What has been gained.}
The chain
\[
 \mathcal A
 \xleftrightarrow{\ \mathcal R_P\ }
 \widehat{\mathcal A}_{\rm HS}
 \xleftrightarrow{\ \mathcal R_B\ }
 \mathcal A_B
 \longleftrightarrow
 \mathcal K
\]
makes the status of Eq.~\eqref{eq:Kframe} explicit.  The function
$\mathcal K$ is the polarized Bargmann kernel of the exact linearized
hard-sphere operator, not an approximation or a separately postulated
generating object.  No Hermite or Sonine order has been selected.  Its Taylor
coefficients recover all Cartesian matrix elements, while
Eq.~\eqref{eq:rotcov} shows that the kernel itself already carries the
rotational symmetry.  The next step is therefore to use that symmetry to
separate angular and radial dynamics.

\section{Exact angular and radial Fock reduction}
\label{sec:angular_radial}

The coherent kernel \eqref{eq:Kframe} is exact but still written in Cartesian
Bargmann variables.  The purpose of this section is to extract from that
single three-dimensional kernel the exact radial operator carried by each
irreducible angular sector.  The construction has two logically separate
parts.  First, the ordinary Cartesian Fock ladder is reorganized into angular
lowest-weight states and a scalar pair ladder that moves only in the radial
direction.  Second, the rotationally invariant coherent kernel is projected
onto those radial towers.  Apart from the particular function
$\mathcal K$, the second step applies to any rotational scalar operator on the
three-dimensional oscillator Fock space.

\subsection{Composite pair creators and the radial tower}\label{sec:su11}

For hard spheres the oscillator number is not conserved, but angular momentum
is.  It is therefore useful to reorganize the Cartesian Hermite states by
three labels,
\[
 \ell=\text{angular tensor rank},\qquad
 m=\text{orientation within that rank},\qquad
 j=\text{radial excitation}.
\]
The familiar kinetic examples are already contained in this notation:
$(j,\ell)=(0,0)$ is density, $(0,1)$ is momentum, $(1,0)$ contains the
scalar energy invariant, $(0,2)$ is the traceless stress sector, and $(1,1)$
contains the heat-flux polynomial.  The label $m$ is passive for a rotational
scalar because all $2\ell+1$ orientations are equivalent.

Rotational covariance and Schur's lemma \cite{Hall2015} give
\begin{equation}
 \langle j,\ell,m|
 \widehat{\mathcal A}_{\rm HS}
 |j',\ell',m'\rangle
 =
 \delta_{\ell\ell'}\delta_{mm'}
 \mathcal A^{(\ell)}_{jj'}.
 \label{eq:block}
\end{equation}
Hence
\begin{equation}
 \widehat{\mathcal A}_{\rm HS}
 =
 \bigoplus_{\ell=0}^\infty
 \mathcal A^{(\ell)}\otimes I_{2\ell+1}.
 \label{eq:blocksum}
\end{equation}
The task is therefore to understand the infinite radial matrix
$\mathcal A^{(\ell)}$ without truncating it.
The radial ladder is already contained in the canonical Fock algebra, and we shall now derive it explicitly.  Define
the rotationally scalar quadratic combinations
\begin{equation}
 \widehat K_+
 =\frac12\hat a_\alpha^\dagger\hat a_\alpha^\dagger,
 \qquad
 \widehat K_-
 =\frac12\hat a_\alpha\hat a_\alpha,
 \qquad
 \widehat K_0
 =\frac12\left(\widehat N+\frac32\right).
 \label{eq:Kpm_abstract}
\end{equation}
They satisfy the standard defining commutation relations of the
$\mathfrak{su}(1,1)$ Lie algebra in this normalization \cite{Perelomov1986},
\begin{equation}
 [\widehat K_0,\widehat K_\pm]=\pm\widehat K_\pm,
 \qquad
 [\widehat K_-,\widehat K_+]=2\widehat K_0.
 \label{eq:su11_commutators}
\end{equation}
Thus the radial algebra is not an additional representation appended to Fock
space: it is generated internally by scalar pairs of the original creation
and annihilation operators.  Since $\widehat K_+$ creates two oscillator
quanta but carries zero angular momentum, it raises radial excitation while
leaving $(\ell,m)$ unchanged.
In the Bargmann realization operators \eqref{eq:Kpm_abstract} become
\begin{equation}
 K_+=\frac12\bm z^2,
 \qquad
 K_-=\frac12\Delta_{\bm z},
 \qquad
 K_0=\frac12\left(
 \bm z\cdot\nabla_{\bm z}+\frac32
 \right),
 \label{eq:Kpm}
\end{equation}
where
$\nabla_{\bm z}=(\partial_{z_1},\partial_{z_2},\partial_{z_3})$ and
$\Delta_{\bm z}=\partial_{z_\alpha}\partial_{z_\alpha}$.
Let $H_{\ell m}(\bm z)$ be a real Fock-orthonormal solid harmonic: it is
homogeneous of degree $\ell$ and harmonic.  Consequently
\begin{equation}
 K_-H_{\ell m}=0,
 \qquad
 K_0H_{\ell m}
 =\left(\frac\ell2+\frac34\right)H_{\ell m}.
 \label{eq:harmonic_lowest_weight}
\end{equation}
Thus $H_{\ell m}$ is the lowest-weight vector, or radial ``vacuum'', for the
fixed-$(\ell,m)$ tower.  It is a vacuum only for the radial pair annihilator;
for $\ell>0$ it is of course not the global Fock vacuum.

Put
\begin{equation}
 k_\ell=\frac\ell2+\frac34,
 \qquad
 \beta_\ell:=2k_\ell=\ell+\frac32.
 \label{eq:k_beta}
\end{equation}
The normalized positive-discrete-series ladder is
\begin{align}
 \widehat K_+|j,\ell,m\rangle
 &=\sqrt{(j+1)(j+\beta_\ell)}
 |j+1,\ell,m\rangle,
 \label{eq:Kplus_ladder}
 \\
 \widehat K_-|j,\ell,m\rangle
 &=\sqrt{j(j+\beta_\ell-1)}
 |j-1,\ell,m\rangle,
 \label{eq:Kminus_ladder}
 \\
 \widehat K_0|j,\ell,m\rangle
 &=\left(j+\frac{\beta_\ell}{2}\right)|j,\ell,m\rangle.
 \label{eq:K0_ladder}
\end{align}
Equivalently, starting from the normalized lowest-weight state,
\begin{equation}
 |j,\ell,m\rangle
 =
 \frac{\widehat K_+^j}{\sqrt{j!(\beta_\ell)_j}}
 |0,\ell,m\rangle .
 \label{eq:radial_ladder_state}
\end{equation}
Here
\begin{equation}
 (a)_j
 =
 \frac{\Gamma(a+j)}{\Gamma(a)}
 =a(a+1)\cdots(a+j-1),
 \qquad
 (a)_0=1,
 \label{eq:pochhammer_def}
\end{equation}
is the rising Pochhammer symbol.  Equation~\eqref{eq:radial_ladder_state} is
the structural origin of the radial index $j$: $j$ counts applications of a
scalar pair creator.  In particular the total oscillator number on this state
is
\begin{equation}
 N=2j+\ell.
 \label{eq:N_2jell}
\end{equation}

Applying the Bargmann realization to
Eq.~\eqref{eq:radial_ladder_state}, using
$\mathcal R_B|0,\ell,m\rangle=H_{\ell m}(\bm z)$ and
$K_+=\bm z^2/2$, gives rather than assumes the normalized radial polynomial:
\begin{equation}
 \mathcal R_B|j,\ell,m\rangle
 =
 \frac{(\bm z^2)^jH_{\ell m}(\bm z)}
 {\sqrt{4^j j!(\ell+3/2)_j}}.
 \label{eq:normalized_state}
\end{equation}
The factor $4^j$ comes from the $1/2$ in the composite creator, while
$j!(\ell+3/2)_j$ is the $\mathfrak{su}(1,1)$ ladder normalization.  Thus the radial polynomial sequence that becomes the conventional
Burnett/Sonine basis after the intertwining constructed below
\cite{ChapmanCowling1970,Kumar1966} is not imposed from outside; it is generated
inside the canonical oscillator Fock space.

There is also a useful kinetic interpretation.  The pair creator
$\widehat K_+$ is not itself the generator of a finite dilation, but it is one
half of the standard radial squeezing generator
$\widehat K_+-\widehat K_-$ \cite{Perelomov1986}.  Exponentiating the latter changes the width of
the Gaussian vacuum, and the resulting squeezed Gaussian expands over the
states $\widehat K_+^j|0\rangle$.  In the scalar sector the first radial
excitation is proportional, in the Hermite realization, to
$|\bm v|^2-3$, the tangent direction to a change of Maxwellian temperature.
An individual $j$ state is not a Maxwellian at a higher temperature.  Its
mean radial energy does increase with $j$, but the more important role of the
tower is to provide the basis in which finite changes of radial scale can be
assembled coherently.  In particular, alternating superpositions of many
pair-created states can squeeze the physical radial energy inward even while
their support moves to large $j$.  This distinction between the location of an
individual basis ket and the scale represented by a coherent radial packet
will become central in the study of global asymptotics below.

At fixed $\ell$ the operators \eqref{eq:Kpm_abstract} therefore realize the
positive discrete series $\mathcal D^+_{k_\ell}$ of $\mathfrak{su}(1,1)$
\cite{Perelomov1986}, and the full Fock space has
the representation decomposition
\begin{equation}
 \mathscr H
 \cong
 \bigoplus_{\ell=0}^\infty
 \mathcal D^+_{k_\ell}\otimes\mathcal V_\ell.
 \label{eq:su11decomp}
\end{equation}
Here $\mathcal V_\ell$ is the $(2\ell+1)$-dimensional irreducible $SO(3)$
space.  The symbol $\cong$ denotes a unitary decomposition of the
representation space; it does not assert an $SO(3)\times \mathfrak{su}(1,1)$ symmetry of
the collision operator.  $SO(3)$ is the hard-sphere symmetry, whereas
$\mathfrak{su}(1,1)$ is the internal radial dynamical algebra that organizes each
angular sector.

\subsection{One-variable radial Bargmann space and extraction of the radial operator}

We now carry out the main task of this section: starting from the already
known coherent kernel $\mathcal K(\bm w,\bm z)$ \eqref{eq:Kframe}, extract the exact radial
operator for a fixed angular momentum $\ell$.  One may choose any magnetic
index $m$ to identify a single radial copy, because Eq.~\eqref{eq:block}
shows that the same operator acts for every $m$; the sum over $m$ is restored
only when the complete three-dimensional kernel is reconstructed.

Set
\begin{equation}
 \xi=\frac{\bm z^2}{2},
 \qquad
 e_j^{(\ell)}(\xi)
 =\frac{\xi^j}{\sqrt{j!(\beta_\ell)_j}}.
 \label{eq:radial_basis_def}
\end{equation}
Then Eq.~\eqref{eq:normalized_state} factorizes as
\begin{equation}
 \mathcal R_B|j,\ell,m\rangle
 =H_{\ell m}(\bm z)e_j^{(\ell)}(\xi).
 \label{eq:radial_basis}
\end{equation}
Because $\mathcal R_B$ is unitary and $H_{\ell m}$ is normalized, the fixed
$(\ell,m)$ Fock inner product induces on the span of the functions
$e_j^{(\ell)}$ the inner product
$\langle e_j^{(\ell)},e_{j'}^{(\ell)}\rangle=\delta_{jj'}$.  We define the
one-variable radial Hilbert space $\mathscr H_\ell^{\rm rad}$ as the completion
of this span in that inherited inner product.  An explicit measure realizing
the same inner product is given below.  Every vector in one fixed $(\ell,m)$ copy is
therefore
\begin{equation}
 \Phi_{\ell m}(\bm z)=H_{\ell m}(\bm z)f(\xi),
 \qquad f\in\mathscr H_\ell^{\rm rad}.
\end{equation}

The reproducing kernel of this radial Hilbert space is, by definition, the
sum over an orthonormal basis \cite{Aronszajn1950}:
\begin{align}
 \mathscr R_\ell(\xi,\bar\eta)
 &:={}
 \sum_{j=0}^\infty
 e_j^{(\ell)}(\xi)\overline{e_j^{(\ell)}(\eta)}
 \nonumber\\
 &=\sum_{j=0}^\infty
 \frac{(\xi\bar\eta)^j}{j!(\beta_\ell)_j}
 ={}_0F_1(;\beta_\ell;\xi\bar\eta).
 \label{eq:radial_RK}
\end{align}
Here
\begin{equation}
 {}_0F_1(;b;s)
 =\sum_{r=0}^{\infty}\frac{s^r}{(b)_r r!}
 \label{eq:0F1_definition_main}
\end{equation}
is the confluent hypergeometric limit function \cite{NIST2010}.  Thus
$\mathscr R_\ell$ is the kernel of the identity operator on the radial space,
not another collision kernel.

A measure realizing the radial inner product is
\begin{equation}
 \dd\mu_\ell(\eta)
 =
 \frac{2}{\pi\Gamma(\beta_\ell)}
 |\eta|^{\beta_\ell-1}
 K_{\beta_\ell-1}(2|\eta|)\dd^2\eta,
 \label{eq:radial_measure}
\end{equation}
where $K_\nu$ is the modified Bessel function of the second kind.  Its moments, obtained from the standard Mellin integral for $K_\nu$ \cite{NIST2010},
are
\begin{equation}
 \int_{\C}\bar\eta^{\,j}\eta^{j'}\dd\mu_\ell(\eta)
 =\delta_{jj'}j!(\beta_\ell)_j.
 \label{eq:measure_moments}
\end{equation}
Thus Eq.~\eqref{eq:measure_moments} verifies directly that the functions
$e_j^{(\ell)}$ are orthonormal in $L^2(\C,\dd\mu_\ell)$, exactly as inherited
from the Fock inner product.  Consequently,
\begin{equation}
 f(\xi)
 =\int_{\C}\mathscr R_\ell(\xi,\bar\eta)f(\eta)\dd\mu_\ell(\eta).
 \label{eq:radial_reproducing_property}
\end{equation}

We next expand the collision kernel in the normalized angular-radial basis.
Using Eq.~\eqref{eq:block} directly in the coherent-state expansion gives
\begin{align}
 \mathcal K(\bm w,\bm z)
 ={}&
 \sum_{\ell=0}^\infty\sum_{m=-\ell}^{\ell}
 \sum_{j,j'=0}^\infty
 \mathcal A^{(\ell)}_{jj'}
 H_{\ell m}(\bm w)H_{\ell m}(\bm z)
 e_j^{(\ell)}\!\left(\frac{\bm w^2}{2}\right)
 e_{j'}^{(\ell)}\!\left(\frac{\bm z^2}{2}\right).
 \label{eq:kernel_basis_expansion}
\end{align}
This equation already shows what has to be separated.  Define the angular
addition kernel
\begin{equation}
 \Pi_\ell(\bm w,\bm z)
 :=\sum_{m=-\ell}^{\ell}H_{\ell m}(\bm w)H_{\ell m}(\bm z)
 \label{eq:Pi_def}
\end{equation}
and the radial operator kernel
\begin{equation}
 \boxed{
 \mathscr A_\ell(\xi,\eta)
 :=
 \sum_{j,j'=0}^\infty
 \mathcal A^{(\ell)}_{jj'}
 e_j^{(\ell)}(\xi)e_{j'}^{(\ell)}(\eta).
 }
 \label{eq:radial_kernel_definition}
\end{equation}
Then \eqref{eq:kernel_basis_expansion} becomes the exact decomposition
\begin{equation}
 \mathcal K(\bm w,\bm z)
 =\sum_{\ell=0}^\infty
 \Pi_\ell(\bm w,\bm z)
 \mathscr A_\ell\!\left(\frac{\bm w^2}{2},\frac{\bm z^2}{2}\right).
 \label{eq:kernel_radial_decomp}
\end{equation}
Thus Eq.~\eqref{eq:kernel_radial_decomp} is not an additional ansatz inferred
from symmetry: it is the coherent-kernel form of the block decomposition
\eqref{eq:block}.

The standard spherical-harmonic addition theorem \cite{NIST2010}, together
with homogeneity of the solid harmonics, implies that $\Pi_\ell$ is a zonal
harmonic and hence is proportional to the degree-$\ell$ Legendre polynomial.
The Fock normalization of the solid harmonics fixes the proportionality
constant, giving the corresponding solid-harmonic addition formula
\begin{equation}
 \Pi_\ell(\bm w,\bm z)
 =
 \frac{(\bm w^2\bm z^2)^{\ell/2}}{(2\ell-1)!!}
 P_\ell\!\left(
 \frac{\bm w\cdot\bm z}{\sqrt{\bm w^2\bm z^2}}
 \right).
 \label{eq:harmonic_kernel}
\end{equation}
Here $P_\ell$ denotes the standard Legendre polynomial of degree $\ell$.
The prefactor is therefore a normalization consequence, not a convention
introduced at the projection stage.  For example,
\begin{equation}
 \Pi_1(\bm w,\bm z)=\bm w\cdot\bm z,
 \qquad
 \Pi_2(\bm w,\bm z)
 =\frac12(\bm w\cdot\bm z)^2-\frac16\bm w^2\bm z^2.
 \label{eq:Pi_low_orders}
\end{equation}

It remains to invert \eqref{eq:kernel_radial_decomp}.  We first work on the
real slice $\bm w,\bm z\in\mathbb R^3$, with nonzero norms.  There a pair
$(\bm w,\bm z)$ modulo a simultaneous $SO(3)$ rotation is completely
specified by the three invariants
\begin{equation}
 \xi=\frac{\bm w^2}{2},\qquad
 \eta=\frac{\bm z^2}{2},\qquad
 \mu=\frac{\bm w\cdot\bm z}{\sqrt{\bm w^2\bm z^2}}.
 \label{eq:orbit_invariants}
\end{equation}
Thus, on this real slice, no generality is lost by choosing the canonical
representative
\begin{equation}
 \bm w=\sqrt{2\xi}\,\bm e_3,
 \qquad
 \bm z=\sqrt{2\eta}\,\bm e_\mu,
 \qquad
 \bm e_\mu:=\mu\bm e_3+\sqrt{1-\mu^2}\bm e_1.
 \label{eq:e_mu_definition}
\end{equation}
For this representative, Eq.~\eqref{eq:harmonic_kernel} gives
\begin{equation}
 \mathcal K\!\left(\sqrt{2\xi}\bm e_3,
                    \sqrt{2\eta}\bm e_\mu\right)
 =\sum_{\ell=0}^\infty
 \frac{2^\ell(\xi\eta)^{\ell/2}}{(2\ell-1)!!}
 P_\ell(\mu)\mathscr A_\ell(\xi,\eta).
 \label{eq:kernel_legendre_series}
\end{equation}
Multiplying by $P_L(\mu)$, integrating from $-1$ to $1$, and using
\begin{equation}
 \int_{-1}^{1}P_L(\mu)P_\ell(\mu)\dd\mu
 =\frac{2}{2\ell+1}\delta_{L\ell}
 \label{eq:Legendre_orthogonality}
\end{equation}
which is the standard Legendre orthogonality relation \cite{NIST2010},
now gives every normalization factor in the projection formula:
\begin{equation}
 \begin{aligned}
 \mathscr A_\ell(\xi,\eta)
 ={}&
 \frac{(2\ell+1)(2\ell-1)!!}
 {2^{\ell+1}(\xi\eta)^{\ell/2}}
 \int_{-1}^1P_\ell(\mu)
 \mathcal K\!\left(
 \sqrt{2\xi}\bm e_3,
 \sqrt{2\eta}\bm e_\mu
 \right)\dd\mu.
 \end{aligned}
 \label{eq:radial_projection}
\end{equation}
This projection has been derived on the real slice $\xi,\eta>0$ and
$-1\leq\mu\leq1$.  Since the coherent kernel is entire in the polarized
Bargmann variables and its angular components are fixed by their homogeneous
Taylor coefficients, the identity extends uniquely to the full complex
Bargmann realization by analytic continuation (equivalently, polarization).
The apparent factor $(\xi\eta)^{-\ell/2}$ is harmless: the $\ell$th
Legendre coefficient of an analytic rotationally invariant Bargmann kernel
starts at angular degree $\ell$ and therefore contains the compensating
factor $(\xi\eta)^{\ell/2}$.  The quotient extends analytically to the axes.

Finally, inserting the explicit radial basis into
Eq.~\eqref{eq:radial_kernel_definition} gives
\begin{equation}
 \boxed{
 \mathscr A_\ell(\xi,\eta)
 =\sum_{j,j'=0}^\infty
 \mathcal A^{(\ell)}_{jj'}
 \frac{\xi^j\eta^{j'}}
 {\sqrt{j!(\beta_\ell)_j\,j'!(\beta_\ell)_{j'}}}.
 }
 \label{eq:radial_generating}
\end{equation}
Thus the radial matrix elements are the normalized Taylor coefficients of the
exact kernel.  It is useful to record both coefficient and derivative forms:
\begin{align}
 \mathcal A^{(\ell)}_{jj'}
 &=\sqrt{j!(\beta_\ell)_j\,j'!(\beta_\ell)_{j'}}
 [\xi^j\eta^{j'}]\mathscr A_\ell(\xi,\eta)
 \label{eq:coefficient_extraction_coefficient}
 \\
 &=\sqrt{\frac{(\beta_\ell)_j(\beta_\ell)_{j'}}{j!j'!}}
 \left.
 \partial_\xi^j\partial_\eta^{j'}\mathscr A_\ell(\xi,\eta)
 \right|_{\xi=\eta=0}.
 \label{eq:coefficient_extraction}
\end{align}
At this stage $\mathcal A^{(\ell)}_{jj'}$ is simply the fixed-$\ell$ Fock
matrix.  Its
identification with the conventional normalized Burnett/Sonine collision
matrix will be made only after the exact physical radial intertwiner is
established in the next section.

On the natural operator domain, the radial operator itself acts as
\begin{equation}
 (\mathfrak A_\ell f)(\xi)
 =\int_{\C}\mathscr A_\ell(\xi,\bar\eta)f(\eta)\dd\mu_\ell(\eta).
 \label{eq:radial_action}
\end{equation}
Comparing Eq.~\eqref{eq:radial_reproducing_property} to \eqref{eq:radial_action}, we recognize that the role of the reproducing kernel $\mathscr R_\ell$ and the collision kernel
$\mathscr A_\ell$ is sharply distinct: the former represents the
identity on the radial Hilbert space, while the latter represents the exact
hard-sphere radial collision operator.

\section{Jacobi functional calculus and Bargmann--Burnett intertwining}
\label{sec:jacobi_intertwining}

The exact radial kernel is complete, but by itself it does not yet separate
the easy and difficult parts of the hard-sphere operator.  The classical
loss--gain decomposition suggests the next question: can the unbounded
collision-frequency part be recognized as a simple operator in the radial
Fock variables?  The answer is yes.  The physical scalar $v^2/2$ becomes a
canonical Jacobi operator, so the whole collision frequency becomes a
spectral function of that Jacobi operator.  This is the step that turns the
exact radialization into a useful spectral representation.

\subsection{The exact radial loss operator}

The loss part is simpler than the gain part and should be kept that way.  We
first derive the physical collision frequency and only then translate the
scalar variable $x=|\bm v|^2/2$ into the radial Fock language.

For hard spheres, the loss term of the linearized operator is multiplication
by the equilibrium mean relative speed,
\begin{equation}
 \nu(\bm v)
 =\kappa\pi\int_{\R^3}|\bm u-\bm v|M(\bm u)\dd\bm u.
 \label{eq:nu_mean_relative_speed}
\end{equation}
This is the loss contribution obtained from the Carleman linearization of
Eq.~\eqref{eq:carleman}.  Let $r=|\bm v|$, write $s=|\bm u|$, and choose the
polar axis along $\bm v$.  If $\mu$ is the cosine of the polar angle, then
\begin{equation}
 \int_{-1}^{1}\sqrt{r^2+s^2-2rs\mu}\dd\mu
 =\frac{(r+s)^3-|r-s|^3}{3rs}.
 \label{eq:relative_speed_angular_integral}
\end{equation}
Splitting the remaining Gaussian radial integral at $s=r$ and introducing the
physical radial-energy variable
\begin{equation}
 x:=\frac{r^2}{2}
\end{equation}
gives directly
\begin{equation}
 \nu(\bm v)=\kappa\pi G(x),
 \qquad
 G(x)=
 \sqrt{\frac2\pi}\e^{-x}
 +\left(\sqrt{2x}+\frac1{\sqrt{2x}}\right)
 \operatorname{erf}(\sqrt x).
 \label{eq:G_reintroduced}
\end{equation}
The value at $x=0$ is understood by continuity.  The function $G$ is the mean
distance from a fixed velocity to a centered Gaussian velocity.  The map
$\bm v\mapsto\int|\bm u-\bm v|M(\bm u)\dd\bm u$ is convex as an
average of the convex norm and is radial by isotropy; hence its minimum is at
$\bm v=0$, i.e. at $x=0$.  This gives the global minimum of $G$ used below.

Let $M_x$ denote multiplication by the physical radial-energy variable
$x=|\bm v|^2/2$ in the Hermite velocity realization,
$(M_x\phi)(\bm v)=(|\bm v|^2/2)\phi(\bm v)$.  We now represent this operator
in Fock space.  Since
multiplication by $v_\alpha$ is intertwined with
$\hat a_\alpha^\dagger+\hat a_\alpha$, one has
\begin{equation}
 \widehat X
 :=\mathcal R_P^{-1}M_x\mathcal R_P
 =\frac12
 (\hat a_\alpha^\dagger+\hat a_\alpha)
 (\hat a_\alpha^\dagger+\hat a_\alpha)
 =\widehat K_++\widehat K_-+2\widehat K_0.
 \label{eq:X_Fock}
\end{equation}
The full action on a radial ket now follows immediately from the ladder
relations \eqref{eq:Kplus_ladder}--\eqref{eq:K0_ladder}:
\begin{equation}
 \begin{aligned}
 \widehat X|j,\ell,m\rangle
 ={}&\sqrt{j(j+\beta_\ell-1)}\,|j-1,\ell,m\rangle
 +(2j+\beta_\ell)|j,\ell,m\rangle
+\sqrt{(j+1)(j+\beta_\ell)}\,|j+1,\ell,m\rangle .
 \end{aligned}
 \label{eq:J_action_full}
\end{equation}
Thus the restriction of $\widehat X$ to a fixed angular sector is the
self-adjoint Jacobi operator $J_\ell$ with entries
\begin{align}
 (J_\ell)_{jj}&=2j+\beta_\ell,
 \label{eq:Jdiag}
 \\
 (J_\ell)_{j,j+1}&=\sqrt{(j+1)(j+\beta_\ell)},
 \label{eq:Joff}
 \\
 (J_\ell)_{j+1,j}&=\sqrt{(j+1)(j+\beta_\ell)}.
\end{align}

Nothing special-function-theoretic
has been assumed here: $J_\ell$ is simply the matrix of physical radial
energy $x=|\bm v|^2/2$ in the normalized pair-created tower.

Let $\widehat X_\ell$ denote the restriction of $\widehat X$ to one
fixed $(\ell,m)$ radial module.  The exact Fock-space loss block is
\begin{equation}
 \widehat{\mathcal N}_\ell
 :=\kappa\pi G(\widehat X_\ell),
 \label{eq:loss_Fock}
\end{equation}
and its matrix in the normalized radial ket basis is
\begin{equation}
 \mathcal N^{(\ell)}=\kappa\pi G(J_\ell).
 \label{eq:loss_J}
\end{equation}
Here $G(J_\ell)$ denotes the spectral functional calculus of the self-adjoint
Jacobi operator $J_\ell$ \cite{ReedSimon1980}.  The Bargmann--Burnett intertwining proved below makes
this notation concrete: $J_\ell$ becomes multiplication by $x$, and
$G(J_\ell)$ becomes multiplication by the same scalar function $G(x)$, on its
natural domain.
For comparison with the one-variable kernel derived in
Sec.~\ref{sec:angular_radial}, restrict the three-dimensional Bargmann operators
$K_\pm,K_0$ of Eq.~\eqref{eq:Kpm} to functions
$H_{\ell m}(\bm z)f(\xi)$.  They induce the explicitly $\ell$-dependent
one-variable operators $K_\pm^{(\ell)},K_0^{(\ell)}$ through
\begin{equation}
 K_\sigma[H_{\ell m}(\bm z)f(\xi)]
 =H_{\ell m}(\bm z)(K_\sigma^{(\ell)}f)(\xi),
 \qquad
 \sigma\in\{+,-,0\}.
 \label{eq:Ksigma_radial_restriction}
\end{equation}
namely
\begin{equation}
 K_+^{(\ell)}=\xi,
 \qquad
 K_-^{(\ell)}=\xi\partial_\xi^2+\beta_\ell\partial_\xi,
 \qquad
 K_0^{(\ell)}=\xi\partial_\xi+\frac{\beta_\ell}{2}.
 \label{eq:Kpm_radial_one_variable}
\end{equation}
Thus these are not new generators but the one-variable radial parts of the
operators in Eq.~\eqref{eq:Kpm}; the $\ell$ dependence enters through
$\beta_\ell$.  Hence
\begin{equation}
 \mathcal J_\ell
 =\xi\partial_\xi^2
 +(2\xi+\beta_\ell)\partial_\xi
 +\xi+\beta_\ell.
 \label{eq:Jdiff}
\end{equation}
Consequently the complete fixed-$\ell$ collision operator can be written
\begin{equation}
 \boxed{
 \mathfrak A_\ell
 =\kappa\pi G(\mathcal J_\ell)-\mathfrak C_\ell,
 }
 \label{eq:radial_operator}
\end{equation}
on $\Dom(\mathfrak A_\ell)=\Dom(G(\mathcal J_\ell))$.  The compact operator
$\mathfrak C_\ell$ is the fixed-$\ell$ gain block.  The Bargmann--Burnett
intertwiner below will make the functional calculus completely concrete by
turning $\mathcal J_\ell$ into multiplication by $x$.

The threshold and the two asymptotic forms needed later are
\begin{equation}
 G(0)=2\sqrt{\frac2\pi},
 \qquad
 \nu_{\min}=2\kappa\sqrt{2\pi},
 \label{eq:numin}
\end{equation}
\begin{align}
 G(x)
 &=2\sqrt{\frac2\pi}
 +\frac23\sqrt{\frac2\pi}\,x+O(x^2),
 \qquad x\downarrow0,
 \label{eq:Gsmall}
 \\
 G(x)
 &=\sqrt{2x}+\frac1{\sqrt{2x}}
 +O(\e^{-x}x^{-2}),
 \qquad x\to\infty.
 \label{eq:Glarge}
\end{align}

\subsection{Exact Bargmann--Burnett transform}

The one-variable Bargmann operator is useful analytically, but it should not
be mistaken for a new kinetic model.  We therefore connect it explicitly to
the classical physical radial coordinate.  The transform below is unitary and
coefficient preserving: it shows that the analytic Bargmann equation and the
Laguerre/Burnett equation are two exact realizations of the same fixed-$\ell$
operator.

Let $L_j^{(\alpha)}(x)$ denote the generalized Laguerre polynomial of degree
$j$ and parameter $\alpha$.  In the physical radial coordinate $x=v^2/2$,
define
\begin{equation}
 u_j^{(\ell)}(x)
 =
 (-1)^j
 \sqrt{\frac{j!}{(\beta_\ell)_j}}\,
 L_j^{(\beta_\ell-1)}(x).
 \label{eq:burnett_radial}
\end{equation}
These functions form a complete orthonormal basis with respect to the
Laguerre weight \cite{Chihara1978,NIST2010},
\begin{equation}
 \dd\rho_\ell(x)
 =
 \frac{x^{\beta_\ell-1}\e^{-x}}{\Gamma(\beta_\ell)}\dd x.
 \label{eq:rho}
\end{equation}
The coefficientwise transform from the radial Bargmann basis to the Burnett
basis has kernel, using the Laguerre generating function \cite{NIST2010},
\begin{equation}
 \begin{aligned}
 \mathscr B_\ell(\xi,x)
 &:=
 \sum_{j=0}^\infty
 e_j^{(\ell)}(\xi)u_j^{(\ell)}(x)
 \\
 &=
 \boxed{
 \e^{-\xi}\,
 {}_0F_1(; \beta_\ell;x\xi)
 }.
 \end{aligned}
 \label{eq:Bkernel}
\end{equation}

\begin{proposition}[Radial intertwining]
Let $\mathcal U_\ell$ be the unitary transform defined by
$\mathcal U_\ell e_j^{(\ell)}=u_j^{(\ell)}$.  In the Burnett radial
realization the same $M_x$ acts as $(M_xg)(x)=xg(x)$.  Then
\begin{equation}
 \mathcal U_\ell\mathcal J_\ell
 =
 M_x\mathcal U_\ell,
 \label{eq:Jintertwine}
\end{equation}
and, for every bounded continuous function $h$ (and, more generally, for a
measurable $h$ on its natural spectral domain),
\begin{equation}
 \mathcal U_\ell h(\mathcal J_\ell)\mathcal U_\ell^{-1}
 =M_{h(x)},
 \label{eq:functional_intertwine}
\end{equation}
where $M_{h(x)}$ denotes multiplication by $h(x)$.  In particular, for the
unbounded hard-sphere function $G$ the equality is understood on the domain
specified in Appendix~\ref{app:operator_domain}:
\begin{equation}
 \mathcal U_\ell G(\mathcal J_\ell)\mathcal U_\ell^{-1}
 =M_{G(x)}.
\end{equation}
\end{proposition}

\begin{proof}
The generalized Laguerre recurrence \cite{Chihara1978,NIST2010}, with the normalization and sign in
Eq.~\eqref{eq:burnett_radial}, reads
\begin{equation}
 \begin{aligned}
 x u_j^{(\ell)}(x)
 ={}&\sqrt{j(j+\beta_\ell-1)}\,u_{j-1}^{(\ell)}(x)
 +(2j+\beta_\ell)u_j^{(\ell)}(x)
 +\sqrt{(j+1)(j+\beta_\ell)}\,u_{j+1}^{(\ell)}(x).
 \end{aligned}
 \label{eq:Laguerre_Jacobi_recurrence}
\end{equation}
This is exactly the ket action \eqref{eq:J_action_full}.  Therefore
$\mathcal U_\ell\mathcal J_\ell=M_x\mathcal U_\ell$ on the polynomial core.
The kernel form \eqref{eq:Bkernel} is the generating-function version of the
same statement; equivalently,
\begin{equation}
 \mathcal J_{\ell,\xi}
 \left[\e^{-\xi}{}_0F_1(;\beta_\ell;x\xi)\right]
 =x\left[\e^{-\xi}{}_0F_1(;\beta_\ell;x\xi)\right].
\end{equation}
Since $\mathcal U_\ell$ is unitary, the functional-calculus statement follows
from the spectral theorem, with the natural domain qualification for
unbounded $h$; see, for example, \cite{ReedSimon1980}.
\end{proof}

The same map intertwines the complete fixed-$\ell$ collision operator, not
only its loss part, because both radial operators are realizations of the same
abstract Fock block.  To make this explicit, let
$\mathcal R_{B,\ell}$ denote the restriction of the Bargmann realization to
one fixed $(\ell,m)$ radial module, and let $\mathcal R_{{\rm Burn},\ell}$
map the same abstract kets $|j,\ell,m\rangle$ to the normalized Burnett radial
functions $u_j^{(\ell)}$.  By construction,
\begin{equation}
 \mathcal U_\ell
 =\mathcal R_{{\rm Burn},\ell}\mathcal R_{B,\ell}^{-1}.
 \label{eq:U_as_change_of_realization}
\end{equation}
If $\widehat{\mathcal A}_\ell$ is the abstract fixed-$\ell$ collision block,
with gain part $\widehat{\mathcal C}_\ell$ and loss part
$\widehat{\mathcal N}_\ell$, then
\begin{equation}
 \mathfrak A_\ell
 =\mathcal R_{B,\ell}\widehat{\mathcal A}_\ell
  \mathcal R_{B,\ell}^{-1},
 \qquad
 \mathfrak A_\ell^{\rm Burn}
 =\mathcal R_{{\rm Burn},\ell}\widehat{\mathcal A}_\ell
  \mathcal R_{{\rm Burn},\ell}^{-1}.
 \label{eq:two_radial_realizations}
\end{equation}
The same convention applies to the gain block:
$\mathfrak C_\ell=\mathcal R_{B,\ell}\widehat{\mathcal C}_\ell
\mathcal R_{B,\ell}^{-1}$; its radial ket matrix is
$\mathcal C^{(\ell)}$.
Therefore realization covariance immediately gives the full intertwining
\begin{equation}
 \boxed{
 \mathcal U_\ell\mathfrak A_\ell
 =\mathfrak A_\ell^{\rm Burn}\mathcal U_\ell.
 }
 \label{eq:full_radial_intertwining}
\end{equation}
Thus the radial Bargmann problem and the conventional Burnett radial problem
are two coordinate realizations of the same fixed-$\ell$ operator.  The
distinction from a Sonine approximation is that no truncation is involved in
the intertwiner.

\section{Burnett/Sonine coordinates and low-order validation}
\label{sec:burnett_validation}

The derivation of the operator is now complete in both radial realizations.
Section~\ref{sec:angular_radial} constructed the exact Bargmann kernel
$\mathscr A_\ell$; Sec.~\ref{sec:jacobi_intertwining} identified the loss as the
spectral function $\kappa\pi G(J_\ell)$ and proved both the loss intertwining
$\mathcal U_\ell\mathcal J_\ell=M_x\mathcal U_\ell$ and the full collision-operator
intertwining
$\mathcal U_\ell\mathfrak A_\ell=\mathfrak A_\ell^{\rm Burn}\mathcal U_\ell$.
The first makes the loss functional calculus explicit; the second states that
Bargmann and Burnett coordinates realize the same complete fixed-$\ell$ block.
Only at this point do
we compare with the classical Burnett/Sonine language; for the traditional
Sonine-polynomial treatment of the Chapman--Enskog transport problem see, in
particular, Chapman and Cowling \cite{ChapmanCowling1970}.  The comparison is
therefore downstream of the exact construction and can be used as a direct
normalization check.

Because $\mathcal U_\ell e_j^{(\ell)}=u_j^{(\ell)}$, the matrix of the exact
operator is coefficient preserving.  We write
\begin{equation}
 \langle f,g\rangle_{B,\ell}
 :=\int_{\C}\overline{f(\xi)}g(\xi)\dd\mu_\ell(\xi),
 \qquad
 \langle f,g\rangle_{{\rm Burn},\ell}
 :=\int_0^\infty\overline{f(x)}g(x)\dd\rho_\ell(x)
 \label{eq:radial_inner_products}
\end{equation}
for the Bargmann and Burnett radial inner products, respectively.  Then
\begin{equation}
 \boxed{
 \mathcal A^{(\ell)}_{jj'}
 =\langle e_j^{(\ell)},\mathfrak A_\ell e_{j'}^{(\ell)}\rangle_{B,\ell}
 =\langle u_j^{(\ell)},\mathfrak A_\ell^{\rm Burn}u_{j'}^{(\ell)}\rangle_{{\rm Burn},\ell}.
 }
 \label{eq:Burnett_same_matrix}
\end{equation}
Thus the Taylor coefficients already extracted by
Eqs.~\eqref{eq:coefficient_extraction_coefficient}--\eqref{eq:coefficient_extraction}
are exactly the normalized Burnett/Sonine collision brackets.

\subsection{Stress block directly from the exact radial kernel}

The $\ell=2$ sector gives a useful fully explicit check because its first
state is the shear-stress mode.  Here $\beta_2=7/2$ and
\begin{equation}
 e_0^{(2)}=1,
 \qquad
 e_1^{(2)}(\eta)=\frac{\eta}{\sqrt{7/2}},
 \qquad
 e_2^{(2)}(\eta)=\frac{\eta^2}{\sqrt{2!(7/2)_2}}.
 \label{eq:l2_radial_basis_low}
\end{equation}
In the original three-dimensional Bargmann variables these are
\begin{equation}
 H_{2m}(\bm z),
 \qquad
 \frac{\bm z^2H_{2m}(\bm z)}{\sqrt{14}},
 \qquad
 \frac{(\bm z^2)^2H_{2m}(\bm z)}{6\sqrt{14}}.
 \label{eq:l2_full_states_low}
\end{equation}
The angular addition kernel is
\begin{equation}
 \Pi_2(\bm w,\bm z)
 =\frac12(\bm w\cdot\bm z)^2
 -\frac16\bm w^2\bm z^2.
 \label{eq:P2_again}
\end{equation}

Projecting the exact coherent kernel by Eq.~\eqref{eq:radial_projection} and
expanding the resulting one-variable kernel gives, through the orders needed
for the first two radial levels,
\begin{equation}
 \frac{\mathscr A_2(\xi,\eta)}{\kappa\sqrt\pi}
 =\frac{16}{5}
 +\frac{8}{35}(\xi+\eta)
 +\frac{164}{147}\xi\eta
 -\frac{2}{315}(\xi^2+\eta^2)
 +\cdots .
 \label{eq:A2_radial_low_expansion}
\end{equation}
The comparison with the normalized basis is now direct.  Since
$e_0^{(2)}(\xi)e_1^{(2)}(\eta)=\eta/\sqrt{7/2}$, matching the coefficient of
$\eta$ in Eq.~\eqref{eq:A2_radial_low_expansion} with the general expansion
\eqref{eq:radial_generating} gives
\begin{equation}
 \mathcal A^{(2)}_{01}
 =\kappa\sqrt\pi\frac{8}{35}\sqrt{\frac72}
 =\frac{4\kappa\sqrt{14\pi}}{35}.
\end{equation}
Similarly the $\eta^2$ term gives
\begin{equation}
 \mathcal A^{(2)}_{02}
 =-\kappa\sqrt\pi\frac{2}{315}
 \sqrt{2!\left(\frac72\right)_2}
 =-\frac{\kappa\sqrt{14\pi}}{105}.
\end{equation}
The constant and $\xi\eta$ terms give the diagonal entries.  Altogether,
\begin{equation}
 \boxed{
 \frac{1}{\kappa\sqrt\pi}
 \begin{pmatrix}
 \mathcal A^{(2)}_{00}&\mathcal A^{(2)}_{01}\\
 \mathcal A^{(2)}_{10}&\mathcal A^{(2)}_{11}
 \end{pmatrix}
 =
 \begin{pmatrix}
 \dfrac{16}{5}&\dfrac{4\sqrt{14}}{35}\\[2mm]
 \dfrac{4\sqrt{14}}{35}&\dfrac{82}{21}
 \end{pmatrix},
 }
 \label{eq:l2_two_state_block}
\end{equation}
and the next first-row coefficient is
\begin{equation}
 \mathcal A^{(2)}_{02}
 =\mathcal A^{(2)}_{20}
 =-\frac{\kappa\sqrt{14\pi}}{105}.
 \label{eq:A202}
\end{equation}
The nonzero off-diagonal entries show directly that hard spheres do not
preserve oscillator number:
\begin{equation}
 [\widehat N,\widehat{\mathcal A}_{\rm HS}]\ne0.
\end{equation}

The same calculation can be displayed directly at the level of homogeneous
pieces of the original coherent kernel.  Since
$\mathcal K_{22}=\mathcal A^{(2)}_{00}\Pi_2$,
\begin{equation}
 \mathcal K_{22}
 =\frac{16}{5}\kappa\sqrt\pi\,\Pi_2.
\end{equation}
At bidegree $(2,4)$,
\begin{equation}
 \mathcal K_{24}
 =\frac{4\kappa\sqrt\pi}{35}\bm z^2\Pi_2
 =\mathcal A^{(2)}_{01}\frac{\bm z^2}{\sqrt{14}}\Pi_2,
\end{equation}
and at bidegree $(2,6)$,
\begin{equation}
 \mathcal K_{26}
 =-\frac{\kappa\sqrt\pi}{630}(\bm z^2)^2\Pi_2
 =\mathcal A^{(2)}_{02}
 \frac{(\bm z^2)^2}{6\sqrt{14}}\Pi_2.
\end{equation}
This makes explicit how the three-dimensional coherent kernel, the radial
Taylor kernel, and the conventional Burnett matrix are the same information
in three coordinate forms.

\subsection{Heat block and finite Sonine compressions}

For heat conduction the momentum state $j=0$ in the $\ell=1$ tower is a
collision invariant and is removed.  Since $\beta_1=5/2$, the first two
non-invariant radial basis functions are
\begin{equation}
 e_1^{(1)}(\xi)=\frac{\xi}{\sqrt{5/2}},
 \qquad
 e_2^{(1)}(\xi)=\frac{\xi^2}{\sqrt{2!(5/2)_2}}.
 \label{eq:l1_radial_basis_low}
\end{equation}
The $j=0$ row and column of $\mathscr A_1$ vanish by momentum conservation.
Projecting the same exact coherent kernel and expanding the first
non-invariant orders gives
\begin{equation}
 \frac{\mathscr A_1(\xi,\eta)}{\kappa\sqrt\pi}
 =\frac{64}{75}\xi\eta
 +\frac{32}{525}(\xi\eta^2+\xi^2\eta)
 +\frac{48}{245}\xi^2\eta^2+\cdots .
 \label{eq:A1_radial_low_expansion}
\end{equation}
For instance, the $\xi\eta^2$ coefficient satisfies
\begin{equation}
 \frac{32}{525}
 =\frac{1}{\kappa\sqrt\pi}
 \frac{\mathcal A^{(1)}_{12}}
 {\sqrt{(5/2)\,2!(5/2)_2}},
\end{equation}
which gives $\mathcal A^{(1)}_{12}/(\kappa\sqrt\pi)=16\sqrt7/105$.
Reading the other two coefficients against the same normalized basis gives
\begin{equation}
 \boxed{
 \frac{1}{\kappa\sqrt\pi}
 \begin{pmatrix}
 \mathcal A^{(1)}_{11}&\mathcal A^{(1)}_{12}\\
 \mathcal A^{(1)}_{21}&\mathcal A^{(1)}_{22}
 \end{pmatrix}
 =
 \begin{pmatrix}
 \dfrac{32}{15}&\dfrac{16\sqrt7}{105}\\[2mm]
 \dfrac{16\sqrt7}{105}&\dfrac{24}{7}
 \end{pmatrix}.
 }
 \label{eq:l1_two_state_block}
\end{equation}
Again, these entries are normalized coefficients of the exact radial kernel,
not separately defined velocity-space collision integrals.

A finite Sonine approximation is now simply a Galerkin compression of the
exact operator.  Define
\begin{equation}
 \mathsf P_N^{(\ell)}
 =\sum_{j=0}^{N}|e_j^{(\ell)}\rangle\langle e_j^{(\ell)}|,
 \qquad
 \mathfrak A_{\ell,N}
 =\mathsf P_N^{(\ell)}\mathfrak A_\ell\mathsf P_N^{(\ell)}.
 \label{eq:finite_compression_sonine}
\end{equation}
For stress the source is $e_0^{(2)}$; for heat one first removes the momentum
invariant and uses the source $e_1^{(1)}$.  Up to the conventional
thermodynamic prefactors,
\begin{align}
 \eta_N&\propto
 \langle e_0^{(2)},\mathfrak A_{2,N}^{-1}e_0^{(2)}\rangle,
 \label{eq:eta_N_practical}
 \\
 \lambda_{{\rm th},N}&\propto
 \langle e_1^{(1)},(\mathfrak A_{1,N}^{\perp})^{-1}e_1^{(1)}\rangle.
 \label{eq:lambda_N_practical}
\end{align}
For a symmetric two-state collision matrix
\begin{equation}
 A_2=\begin{pmatrix}a&b\\b&c\end{pmatrix},
\end{equation}
the one-state inverse-collision response is $1/a$, whereas the two-state
response is $(A_2^{-1})_{11}=c/(ac-b^2)$.  Their ratio is therefore
\begin{equation}
 \frac{c/(ac-b^2)}{1/a}=\frac{ac}{ac-b^2}.
 \label{eq:two_state_inverse_ratio}
\end{equation}
Substitution of the stress and heat blocks
\eqref{eq:l2_two_state_block} and \eqref{eq:l1_two_state_block} then gives
\begin{equation}
 \frac{\eta_2}{\eta_1}=\frac{205}{202},
 \qquad
 \frac{\lambda_{{\rm th},2}}{\lambda_{{\rm th},1}}=\frac{45}{44}.
 \label{eq:second_sonine_ratios}
\end{equation}

These are the classical second-Sonine corrections.  Increasing the radial
compression gives
\begin{equation}
\begin{array}{c|ccccc}
\text{active radial states} &1&2&3&4&5\\ \hline
\eta_N/\eta_1
&1&1.014852&1.015879&1.016006&1.016028
\end{array}
\label{eq:eta_sequence}
\end{equation}
and
\begin{equation}
\begin{array}{c|cccc}
\text{active radial states} &1&2&3&4\\ \hline
\lambda_{{\rm th},N}/\lambda_{{\rm th},1}
&1&1.022727&1.024818&1.025134
\end{array}
\label{eq:lambda_sequence}
\end{equation}
in agreement with the classical rigid-sphere Sonine sequence
\cite{ReineckeKremer1990,HiemstraEtAl2026}.  The infinite-order benchmarks of
Pekeris and Alterman \cite{PekerisAlterman1957} are
\begin{equation}
 \frac{\eta_\infty}{\eta_1}=1.016034,
 \qquad
 \frac{\lambda_{{\rm th},\infty}}{\lambda_{{\rm th},1}}=1.025218.
 \label{eq:transport_infinite_benchmarks}
\end{equation}
These numerical comparisons validate the normalization and the low-order
coefficient extraction.  They do not define the operator: the exact radial
kernel and its intertwining were already established before any finite
compression was selected.

\paragraph{End of the construction and validation.}
It is useful to summarize what has now been obtained before turning to the
global spectral analysis.  The Carleman representation produced the exact coherent
kernel $\mathcal K(\bm w,\bm z)$.  Rotational projection produced the exact
one-variable kernels $\mathscr A_\ell$.  The physical loss became the spectral
function $\kappa\pi G(J_\ell)$ of the radial-energy Jacobi operator.  The
unitary map $\mathcal U_\ell$ identified the analytic Bargmann and physical
Laguerre/Burnett realizations.  Finally, the conventional Burnett/Sonine
matrices were recovered as coordinate matrices of this already constructed
operator.  No radial truncation has entered any of these steps; truncation
appears only when a finite Galerkin compression such as
Eq.~\eqref{eq:finite_compression_sonine} is chosen.

\section{Global spectral consequences}
\label{sec:global_spectral}

The preceding Burnett/Sonine validation probes only a few low-lying radial
states.  The exact operator now allows a genuinely global question: what can
be said about an entire angular sector without fixing a Sonine depth?  Here the elementary
Boltzmann decomposition itself becomes spectral language:
\begin{equation}
 \boxed{
 \widehat{\mathcal A}_\ell
 =\underbrace{\widehat{\mathcal N}_\ell}_{\text{loss: noncompact}}
 -\underbrace{\widehat{\mathcal C}_\ell}_{\text{gain: compact}},
 \qquad
 \mathcal N^{(\ell)}=\kappa\pi G(J_\ell).
 }
 \label{eq:global_loss_gain_roles}
\end{equation}
The asymmetry is worth stressing.  The gain is the technically demanding part
of the exact Carleman--Bargmann construction because it contains the full
redistribution geometry of collisions.  Spectrally, however, it is compact.
The loss is much simpler in the physical radial realization---it is only
multiplication by $\kappa\pi G(x)$---but it is the noncompact part and therefore
sets the entire essential continuum, including its threshold, and the
large-index geometry.  The gain remains
indispensable for quantitative transport corrections and for any discrete
spectrum below the threshold, but it is asymptotically subordinate in the two
large-index limits developed below.  In this sense the loss sets the global
spectral landscape while the gain modifies it compactly.

This distinction also explains why finite Burnett/Sonine calculations can be
excellent for low-order transport yet incomplete for global spectral
questions.  A finite compression can approximate low collision brackets very
accurately, including the compact gain contribution, but it necessarily turns
the noncompact multiplication operator $G(x)$ into a finite matrix.  It can
therefore approximate selected responses without reproducing, at fixed order,
the continuum mechanism of the loss.

We now use the exact radial form rather than its Taylor truncations.  The
standard hard-sphere gain operator $\mathcal C$ in the weighted velocity
Hilbert space is compact; see, for example,
\cite{Grad1963,Drange1975,Pekeris1963,Dudynski2013}.  Let
$\widehat{\mathcal C}$ denote its unitarily equivalent Fock realization.
Rotational symmetry gives the block decomposition
\begin{equation}
 \widehat{\mathcal C}
 =
 \bigoplus_{\ell=0}^\infty
 \widehat{\mathcal C}_\ell\otimes I_{2\ell+1},
 \label{eq:Cblocks}
\end{equation}
where $\widehat{\mathcal C}_\ell$ is the abstract radial gain block.  Its
matrix in the radial ket basis will be denoted $\mathcal C^{(\ell)}$, while
its one-variable Bargmann realization is $\mathfrak C_\ell$ as in
Eq.~\eqref{eq:radial_operator}.  Compactness and operator norm are preserved
under these unitary changes of realization.

\begin{theorem}[Common essential spectrum]\label{thm:common_essential}
For every $\ell\ge0$,
\begin{equation}
 \boxed{
 \specess(\widehat{\mathcal A}_\ell)
 =\specess(\mathfrak A_\ell)
 =
 [\nu_{\min},\infty).
 }
 \label{eq:ess}
\end{equation}
\end{theorem}

\begin{proof}
Under the Bargmann--Burnett transform,
$\kappa\pi G(\mathcal J_\ell)$ is multiplication by
$\kappa\pi G(x)$ on the Burnett radial space
$L^2(\R_+,\dd\rho_\ell)$, where
$\dd\rho_\ell(x)=x^{\beta_\ell-1}\e^{-x}\dd x/\Gamma(\beta_\ell)$ as in
Eq.~\eqref{eq:rho}.  By the spectral characterization of multiplication
operators \cite{ReedSimon1980}, its spectrum is the essential range of
$\kappa\pi G(x)$.  Since $G$ is continuous, has minimum $G(0)$, and diverges
as $\sqrt{2x}$, this essential range is $[\nu_{\min},\infty)$.  The
radial gain block $\widehat{\mathcal C}_\ell$ is compact (equivalently,
so is $\mathfrak C_\ell$ in Bargmann coordinates), and Weyl's theorem on compact perturbations \cite{ReedSimon1980} leaves
the essential spectrum unchanged.
\end{proof}

\begin{lemma}[Decay of angular gain blocks]
\begin{equation}
 \|\widehat{\mathcal C}_\ell\|\longrightarrow0
 \qquad(\ell\to\infty).
 \label{eq:Cldecay}
\end{equation}
\end{lemma}

\begin{proof}
If not, there would exist $\varepsilon>0$, a sequence
$\ell_n\to\infty$, and normalized radial vectors $f_n$ in mutually orthogonal
angular sectors such that
$\|\widehat{\mathcal C}_{\ell_n}f_n\|\ge\varepsilon$.  Let $\mathsf P_\ell$ be the
orthogonal projection onto the full angular sector of degree $\ell$.  For any
fixed vector $g$ in the Hilbert direct sum,
\begin{equation}
 |\langle g,f_n\rangle|
 =|\langle \mathsf P_{\ell_n}g,f_n\rangle|
 \le \|\mathsf P_{\ell_n}g\|\longrightarrow0,
\end{equation}
because $\sum_{\ell\ge0}\|\mathsf P_\ell g\|^2=\|g\|^2$.  Thus
$f_n\rightharpoonup0$.  By the standard weak-to-strong property of compact
operators \cite{ReedSimon1980}, compactness of $\widehat{\mathcal C}$ then implies
$\|\widehat{\mathcal C}f_n\|\to0$, whereas on the $\ell_n$ sector this norm is
$\|\widehat{\mathcal C}_{\ell_n}f_n\|\ge\varepsilon$, a contradiction.
\end{proof}

As a soft consequence of the preceding two results,
\begin{equation}
 \nu_{\min}-\|\widehat{\mathcal C}_\ell\|
 \le
 \inf\spec(\widehat{\mathcal A}_\ell)
 \le
 \nu_{\min},
 \label{eq:bottom_bounds_soft}
\end{equation}
and hence $\inf\spec(\widehat{\mathcal A}_\ell)\to\nu_{\min}$.  The lower bound follows
from $\widehat{\mathcal N}_\ell\ge\nu_{\min}I$, while the upper bound follows from
the presence of $\nu_{\min}$ in the essential spectrum.  This conclusion is
classical in character: once the loss--gain decomposition, the range of the
collision frequency, and compactness of the gain are known, no Bargmann
machinery is required.

Together with the classical absence of accumulation at the threshold for hard
spheres \cite{Dudynski2013}, the bound is consistent with only finitely many
angular sectors carrying discrete spectrum.  Klaus obtained a much sharper
sector-wise result and reported computer-assisted evidence that $\ell=3$ is
the first sector without discrete eigenvalues \cite{Klaus1976}; we do not
attempt to recover that critical value here.

The purpose of the next two sections is different.  We use the explicit radial
structure to identify normalized Fock states that realize the continuum and
its threshold constructively.  The large-$\ell$ analysis shows why every fixed
radial depth misses the threshold and constructs a family that escapes to the
required radial scale.  The large-$j$ analysis then shows how the same
threshold appears at fixed angular momentum as a genuine radial Weyl sequence.

\section{Large-\texorpdfstring{$\ell$}{l} asymptotics: fixed radial depth and squeezed continuum access}
\label{sec:large_l}

The soft spectral argument above already knows the limiting threshold, but it
does not explain how a state in a high angular sector reaches the low physical
energy $x=0$.  The exact radial representation lets us answer that question
constructively.  There are two sharply different parts to the story.  First,
at every fixed radial depth the loss becomes increasingly diagonal and is
pushed to the large collision scale $O(\sqrt\ell)$.  Second, the exact radial
tower contains states whose radial support escapes with $\ell$ and which
remain at finite, or even vanishing, physical energy.

\subsection{Fixed radial depth: asymptotic diagonalization}

We first compare the exact infinite radial problem with the customary
finite-Sonine viewpoint.

Set
\begin{equation}
 \beta=\beta_\ell=\ell+\frac32,
 \qquad
 J_\ell=\beta I+T_\ell,
 \qquad
 T_\ell:=J_\ell-\beta I.
 \label{eq:T_shift_definition}
\end{equation}
The shifted Jacobi matrix $T_\ell$ is not a new kinetic object.  It isolates
the radial variation about the common $O(\ell)$ angular background
$\beta I$.  On every fixed radial window its diagonal part is $O(1)$ and its
nearest-neighbor part is $O(\sqrt\beta)$, so $T_\ell/\beta$ is small, with
leading size $O(\beta^{-1/2})$.  This is precisely the scale separation used
below to expand $G(\beta I+T_\ell)$ and to identify the first nontrivial
Hermite--Jacobi correction.

Within this section we suppress the fixed passive angular labels and write
$|j\rangle:=|j,\ell,m\rangle$ for the normalized radial ket; thus
$\langle j|T_\ell^n|j'\rangle$ denotes the corresponding radial matrix
element.  For fixed $j$,
\begin{equation}
 (T_\ell)_{jj}=2j,
 \qquad
 (T_\ell)_{j,j+1}
 =
 \sqrt{\beta}\sqrt{j+1}
 \left[1+\frac{j}{2\beta}+O(\beta^{-2})\right].
\end{equation}
The large-$x$ expansion \eqref{eq:Glarge} and a finite-matrix expansion in
powers of $\beta^{-1/2}$ give the following.

\begin{proposition}[Fixed radial indices, loss part]
For fixed $j$,
\begin{equation}
 \frac{\mathcal N^{(\ell)}_{jj}}{\kappa\pi}
 =
 \sqrt{2\beta}
 \left[
 1+
 \frac{3(2j+1)}{8\beta}
 +O(\beta^{-2})
 \right].
 \label{eq:large_l_diag}
\end{equation}
For fixed $d\ge1$,
\begin{equation}
 \frac{\mathcal N^{(\ell)}_{j,j+d}}{\kappa\pi}
 =
 \sqrt2
 \binom{1/2}{d}
 \sqrt{(j+1)_d}\,
 \beta^{(1-d)/2}
 \left[1+O(\beta^{-1})\right].
 \label{eq:large_l_offdiag}
\end{equation}
\end{proposition}

\begin{proof}
The diagonal matrix elements needed to relative order $\beta^{-1}$ are
\begin{equation}
 \langle j|T_\ell|j\rangle=2j,
\end{equation}
and
\begin{equation}
 \langle j|T_\ell^2|j\rangle
 =
 \beta(2j+1)+6j^2.
\end{equation}
For fixed radial indices this calculation is an \emph{entrywise}
asymptotic expansion, not a global operator-norm expansion of the unbounded
Jacobi operator.  Indeed, after compression to any fixed finite radial window,
$T_\ell/\beta=O(\beta^{-1/2})$ in finite-dimensional operator norm, so the
spectral-calculus expansion
\begin{equation}
 \sqrt{\beta I+T_\ell}
 =
 \sqrt\beta
 \left[
 I+\frac{T_\ell}{2\beta}
 -\frac{T_\ell^2}{8\beta^2}+\cdots
 \right]
\end{equation}
is legitimate to any prescribed order on that window.  For a fixed matrix
element, the coefficients through a prescribed order depend on only finitely
many Jacobi paths, so the same coefficients are those of the full operator.
Substitution of the two moments above, together with the leading
$J_\ell^{-1/2}$ contribution in \eqref{eq:Glarge}, yields
\eqref{eq:large_l_diag}.  For an off-diagonal
distance $d$, the first nonzero term occurs in $T_\ell^d$.  The unique minimal
raising path gives
\begin{equation}
 (T_\ell^d)_{j,j+d}
 =
 \prod_{r=0}^{d-1}
 \sqrt{(j+r+1)(\beta+j+r)}
 =
 \beta^{d/2}\sqrt{(j+1)_d}
 [1+O(\beta^{-1})].
\end{equation}
The binomial coefficient in \eqref{eq:large_l_offdiag} follows.
\end{proof}

In particular,
\begin{align}
 \mathcal N^{(\ell)}_{j,j+1}
 &=
 \kappa\pi\sqrt{\frac{j+1}{2}}
 +O(\ell^{-1}),
 \label{eq:l_nn}
 \\
 \mathcal N^{(\ell)}_{j,j+2}
 &=
 -\frac{\kappa\pi}{4\sqrt2}
 \sqrt{(j+1)(j+2)}\,\ell^{-1/2}
 +O(\ell^{-3/2}).
\end{align}

The qualitative content of these estimates is simple.  On any fixed radial
window the common diagonal loss grows as $\sqrt{\ell}$, the nearest-neighbor
coupling stays only $O(1)$, and all couplings at distance $d\ge2$ vanish.
Recall the orthogonal finite-Sonine projector $\mathsf P_N^{(\ell)}$ from
Eq.~\eqref{eq:finite_compression_sonine}; in the present radial ket notation we
abbreviate it to $\mathsf P_N$.  Thus $\mathsf P_N$ projects onto
$\operatorname{span}\{|0\rangle,\ldots,|N\rangle\}$, with $N$ held fixed while
$\ell\to\infty$.  The purpose of the remainder of this subsection is to retain
the first correction to the leading diagonal loss and to determine whether the
compact gain contributes at that $O(1)$ scale.
Thus, relative to its leading scale,
\begin{equation}
 \frac{\mathsf P_N\mathcal N^{(\ell)}\mathsf P_N}
 {\kappa\pi\sqrt{2\beta_\ell}}
 \longrightarrow I_{N+1}
 \qquad(\ell\to\infty)
 \label{eq:large_l_relative_diagonalization}
\end{equation}
for every fixed $N$.  Large angular momentum therefore produces an
asymptotic \emph{collapse toward a diagonal loss operator} at fixed radial
depth.  The $O(1)$ Hermite--Jacobi coupling resolved below is the first
correction to that collapse, not a competitor to the leading $\sqrt\ell$
scale.

Compactness of the full gain operator is now decisive.  From
\eqref{eq:Cldecay}, for every fixed finite radial section,
\begin{equation}
 \|\mathsf P_N\mathcal C^{(\ell)}\mathsf P_N\|\to0.
\end{equation}

Define the $(N+1)\times(N+1)$ symmetric tridiagonal matrix $S_N$ by
\begin{equation}
 (S_N)_{jj}=0,
 \qquad
 (S_N)_{j,j+1}=(S_N)_{j+1,j}
 =\sqrt{\frac{j+1}{2}},
 \qquad 0\le j<N.
 \label{eq:SN_large_l}
\end{equation}
The matrix $S_N$ is the finite Hermite \emph{recurrence matrix}: it represents
multiplication by the argument in the first $N+1$ orthonormal physicists'
Hermite polynomials \cite{Chihara1978}.  We use ``recurrence matrix'' here to
distinguish it from the radial-energy Jacobi operator $J_\ell$.  The preceding
entrywise asymptotics show that $S_N$ is exactly the surviving $O(1)$ correction
to the leading diagonal loss; the next theorem adds the compact gain and makes
that statement precise for the complete collision block.

\begin{theorem}[Universal finite-depth Hermite limit]
Let $\mathsf P_N$ project onto $j=0,\ldots,N$, with $N$ fixed.  Then
\begin{equation}
 \boxed{
 \mathsf P_N\mathcal A^{(\ell)}\mathsf P_N
 =
 \kappa\pi
 \left[
 \sqrt{2\beta_\ell}\,I_{N+1}+S_N
 \right]
 +o(1),
 \qquad
 \ell\to\infty,
 }
 \label{eq:Hermite_limit}
\end{equation}
where $S_N$ is defined in Eq.~\eqref{eq:SN_large_l}.
\end{theorem}

\begin{proof}
Because $N$ is fixed, Eqs.~\eqref{eq:large_l_diag} and
\eqref{eq:large_l_offdiag} apply uniformly to the finitely many entries in the
compressed matrix.  After subtracting
$\kappa\pi\sqrt{2\beta_\ell}\,I_{N+1}$, the diagonal entries are $o(1)$, the
nearest-neighbor entries converge to those of $\kappa\pi S_N$ by
Eq.~\eqref{eq:l_nn}, and every entry at distance at least two is $o(1)$.
Hence
\begin{equation}
 \mathsf P_N\mathcal N^{(\ell)}\mathsf P_N
 =\kappa\pi[\sqrt{2\beta_\ell}\,I_{N+1}+S_N]+o(1)
\end{equation}
in finite-dimensional operator norm.  Equation~\eqref{eq:Cldecay} gives
$\|\mathsf P_N\mathcal C^{(\ell)}\mathsf P_N\|\to0$, so subtracting the gain
block proves Eq.~\eqref{eq:Hermite_limit}.
\end{proof}

Let
$h_{N+1,1},\ldots,h_{N+1,N+1}$
be the zeros of the physicists' Hermite polynomial $H_{N+1}$.  Since $S_N$ is the corresponding orthonormal Hermite recurrence matrix, its eigenvalues are these zeros by the standard Jacobi-matrix/orthogonal-polynomial correspondence \cite{Chihara1978}.  Hence the Ritz values of the
fixed Sonine section satisfy
\begin{equation}
 \lambda^{(N)}_{\ell,k}
 =
 \kappa\pi
 \left[
 \sqrt{2\ell+3}+h_{N+1,k}
 \right]
 +o(1).
 \label{eq:Ritz_large_l}
\end{equation}

The finite-depth result by itself gives only one side of the high-$\ell$
story: it proves that the threshold cannot remain in any fixed Sonine window.
To construct the complementary path it is useful to return first to the
$\mathfrak{su}(1,1)$ radial ladder itself, before invoking any particular coordinate
realization.

\subsection{Squeezed radial packets as nonperturbative spectral probes}
\label{sec:squeezed_radial_path}

The fixed-$j$ asymptotics above follow individual basis kets
$|j,\ell,m\rangle$ while $\ell$ increases.  To probe the continuum one needs a
different object: a normalized coherent packet that uses the entire radial
tower and whose support can migrate through it.  The $\mathfrak{su}(1,1)$ squeeze provides
exactly such a family.  It will be used below in two complementary ways.  At
fixed $\ell$, increasing the squeeze produces a genuine threshold Weyl
sequence.  When $\ell$ also increases, tuning the squeeze against
$\beta_\ell=\ell+3/2$ produces sharply localized packets at any prescribed
finite radial energy and therefore reaches the whole essential band.  The
squeezed family is thus not merely a convenient trial state; it is the
concrete nonperturbative probe that connects the discrete Fock tower to the
continuous spectral variable of the loss operator.

Recall from Sec.~\ref{sec:su11} the scalar pair operators
$\widehat K_\pm$ and $\widehat K_0$.  Since
$\widehat K_+^\dagger=\widehat K_-$, the operator
\begin{equation}
 \widehat{\mathsf S}(r)
 :=\exp\!\left[-r(\widehat K_+-\widehat K_-)\right],
 \qquad r\ge0,
 \label{eq:radial_squeeze_operator}
\end{equation}
is unitary.  It is rotationally scalar, so it does not change $(\ell,m)$.
Starting from the lowest-weight state in the $\ell$th radial module, define
\begin{equation}
 |\Psi_{\ell m}(r)\rangle
 :=\widehat{\mathsf S}(r)|0,\ell,m\rangle .
 \label{eq:squeezed_radial_ket}
\end{equation}
For the high-$\ell$ construction one may choose any magnetic copy, for example
$m=0$; all radial quantities below are independent of $m$.

A standard $SU(1,1)$ disentangling identity rewrites the exponential of the
noncommuting combination $\widehat K_+-\widehat K_-$ as an ordered product of
three exponentials involving separately $\widehat K_+$, $\widehat K_0$, and
$\widehat K_-$ \cite{Perelomov1986}.  On a lowest-weight vector the rightmost
lowering exponential acts trivially because
$\widehat K_-|0,\ell,m\rangle=0$, while the middle exponential contributes only
the scalar lowest-weight normalization.  This leaves a single raising
exponential.  Writing
\begin{equation}
 q:=\tanh r,\qquad 0\le q<1,
 \label{eq:squeeze_disk_parameter}
\end{equation}
one obtains
\begin{equation}
 |\Psi_{\ell m}(r)\rangle
 =(\operatorname{sech} r)^{\beta_\ell}
 \exp(-q\widehat K_+)|0,\ell,m\rangle
 =(1-q^2)^{\beta_\ell/2}
 \exp(-q\widehat K_+)|0,\ell,m\rangle .
 \label{eq:squeeze_disentangled}
\end{equation}
Using Eq.~\eqref{eq:radial_ladder_state}, this is the explicit ket expansion
\begin{equation}
 \boxed{
 |\Psi_{\ell m}(r)\rangle
 =(1-q^2)^{\beta_\ell/2}
 \sum_{j=0}^{\infty}
 (-q)^j
 \sqrt{\frac{(\beta_\ell)_j}{j!}}
 |j,\ell,m\rangle .
 }
 \label{eq:squeezed_ket_expansion}
\end{equation}
Thus $q$ is best viewed as the $\mathfrak{su}(1,1)$ coherent-state, or radial squeezing,
parameter: each power $q^j$ weights the component generated by $j$
applications of the scalar pair creator $\widehat K_+$.  The alternating
phase $(-1)^j$ will become important again in the large-$j$ Toeplitz limit in Section \ref{sec:large_j}.

The radial-index probabilities are therefore
\begin{equation}
 \mathbb P_{\ell,r}(j)
 =|\langle j,\ell,m|\Psi_{\ell m}(r)\rangle|^2
 =\frac{(\beta_\ell)_j}{j!}
 (1-q^2)^{\beta_\ell}q^{2j},
 \label{eq:trial_negative_binomial}
\end{equation}
a negative-binomial distribution.  Its expectation and variance are denoted
by
\begin{equation}
 \langle j\rangle_{\ell,r}
 :=\sum_{j=0}^\infty j\,\mathbb P_{\ell,r}(j),
 \qquad
 \operatorname{var}_{\ell,r}(j)
 :=\sum_{j=0}^\infty
 (j-\langle j\rangle_{\ell,r})^2\mathbb P_{\ell,r}(j),
 \label{eq:squeezed_j_moment_definitions}
\end{equation}
and evaluation of these sums gives
\begin{equation}
 \langle j\rangle_{\ell,r}
 =\frac{\beta_\ell q^2}{1-q^2}
 =\beta_\ell\sinh^2 r,
 \qquad
 \operatorname{var}_{\ell,r}(j)
 =\frac{\beta_\ell q^2}{(1-q^2)^2}.
 \label{eq:squeezed_j_moments}
\end{equation}
Hence increasing the squeeze parameter does not select a single large-$j$
ket; it moves the probability distribution of a coherent radial packet to
larger and larger $j$.

The same group action has a particularly direct meaning for the physical
radial energy.  Recall from Eq.~\eqref{eq:X_Fock} that
\begin{equation}
 \widehat X
 =\widehat K_++\widehat K_-+2\widehat K_0,
 \end{equation}
which is the Fock representative of multiplication by $x=|\bm v|^2/2$.
With $\widehat D:=\widehat K_+-\widehat K_-$, the $SU(1,1)$ commutators give
\begin{equation}
 [\widehat D,\widehat X]=-2\widehat X,
\end{equation}
and therefore
\begin{equation}
 \widehat{\mathsf S}(r)^\dagger\widehat X\widehat{\mathsf S}(r)
 =\e^{-2r}\widehat X .
 \label{eq:squeeze_X_scaling}
\end{equation}
Since
$\langle0,\ell,m|\widehat X|0,\ell,m\rangle=\beta_\ell$ and
$\operatorname{var}_{0,\ell}(X)=\beta_\ell$, Eq.~\eqref{eq:squeeze_X_scaling}
immediately yields
\begin{equation}
 \boxed{
 \langle X\rangle_{\ell,r}
 =\beta_\ell\e^{-2r},
 \qquad
 \operatorname{var}_{\ell,r}(X)
 =\beta_\ell\e^{-4r}.
 }
 \label{eq:squeezed_X_moments}
\end{equation}
Thus the same state that moves outward in the pair-created index $j$ moves
\emph{inward} in physical radial energy.  This is possible because $X$ is not
diagonal in the $j$ basis: the alternating coherent superposition in
Eq.~\eqref{eq:squeezed_ket_expansion} produces cancellation between the
$2\widehat K_0$ and $\widehat K_++\widehat K_-$ contributions.

For comparison with the Burnett realization, introduce
\begin{equation}
 t:=\frac{q}{1-q}=\frac{\e^{2r}-1}{2},
 \qquad
 q=\frac{t}{1+t},
 \qquad
 1+2t=\e^{2r}.
 \label{eq:t_r_q_relation}
\end{equation}
Using the Laguerre generating function \cite{NIST2010} together with the normalized basis
$u_j^{(\ell)}$, the ket \eqref{eq:squeezed_ket_expansion} is represented in
$L^2(\R_+,\dd\rho_\ell)$ by
\begin{equation}
 \psi_{\ell,t}(x)
 :=(1+2t)^{\beta_\ell/2}\e^{-tx}.
 \label{eq:threshold_trial_state}
\end{equation}
Indeed, the two factors $(-1)^j$ in
Eq.~\eqref{eq:squeezed_ket_expansion} and in the convention for
$u_j^{(\ell)}$ cancel, and
\begin{equation}
 \sum_{j=0}^{\infty}q^jL_j^{(\beta_\ell-1)}(x)
 =(1-q)^{-\beta_\ell}
 \exp\!\left[-\frac{qx}{1-q}\right].
\end{equation}
Thus Eq.~\eqref{eq:threshold_trial_state} is not an ad hoc trial function: it
is precisely the Burnett-coordinate image of the unitary squeezed ket
\eqref{eq:squeezed_radial_ket}.  Its inverse Bargmann image is correspondingly
\begin{equation}
 \mathcal U_\ell^{-1}\psi_{\ell,t}(\xi)
 =(1-q^2)^{\beta_\ell/2}\e^{-q\xi}
 =\frac{(1+2t)^{\beta_\ell/2}}{(1+t)^{\beta_\ell}}
 \exp\!\left[-\frac{t}{1+t}\xi\right].
 \label{eq:trial_Bargmann_state}
\end{equation}
Finally, Eq.~\eqref{eq:squeezed_X_moments} becomes
\begin{equation}
 \langle x\rangle_{\ell,t}
 =\frac{\beta_\ell}{1+2t},
 \qquad
 \operatorname{var}_{\ell,t}(x)
 =\frac{\beta_\ell}{(1+2t)^2},
 \label{eq:trial_x_moments}
\end{equation}
in agreement with the Gamma law
$|\psi_{\ell,t}|^2\dd\rho_\ell$ of shape $\beta_\ell$ and rate $1+2t$.

\subsection{Constructive access to the full essential band}
\label{sec:fock_essential_continuum}

For readers accustomed to finite Burnett/Sonine matrices, one point is worth
making explicit before specializing to the threshold.  The essential spectrum
is not external to the Fock representation.  A finite Sonine compression is a
finite matrix and therefore has only discrete spectrum, but the untruncated
radial Fock tower is infinite.  Through the exact Bargmann--Burnett
intertwiner, its Jacobi matrix $J_\ell$ is unitarily equivalent to
multiplication by the continuous variable $x\in[0,\infty)$.  The continuum
of the loss operator is therefore encoded in coherent packets whose Fock
coefficients escape every fixed radial window.

Formally, the Laguerre recurrence may be summarized by the generalized radial
ket
\begin{equation}
 |x;\ell,m\rangle_{\rm gen}
 :=\sum_{j=0}^\infty
 u_j^{(\ell)}(x)|j,\ell,m\rangle,
 \label{eq:generalized_J_eigenket}
\end{equation}
which is not normalizable in Fock space but satisfies, in the generalized
spectral sense,
\begin{equation}
 J_\ell|x;\ell,m\rangle_{\rm gen}
 =x|x;\ell,m\rangle_{\rm gen}.
\end{equation}
Thus every physical radial energy $x$ is already represented by the infinite
Jacobi/Fock problem, even though no single finite Sonine vector realizes it.

The squeezed family makes this continuum access constructive with ordinary
normalized Fock vectors.  We shall use Weyl's criterion in its singular-sequence
form \cite{ReedSimon1980}: for a self-adjoint operator $A$, a normalized sequence
$\psi_n\in\Dom(A)$ satisfying
\begin{equation}
 \psi_n\rightharpoonup0,
 \qquad
 \|(A-\lambda)\psi_n\|\longrightarrow0
 \label{eq:Weyl_criterion_reminder}
\end{equation}
is called a Weyl (or singular) sequence at $\lambda$, and its existence is
equivalent to $\lambda\in\specess(A)$.  The two limits in
Eq.~\eqref{eq:Weyl_criterion_reminder} are exactly what the squeezed packets
will provide.

\begin{theorem}[Squeezed Weyl sequences throughout the essential band]
Fix $x_0>0$ and, for all sufficiently large $\ell$, choose $r_\ell(x_0)$ by
\begin{equation}
 \e^{-2r_\ell(x_0)}=\frac{x_0}{\beta_\ell}.
 \label{eq:continuum_squeeze_choice}
\end{equation}
Then
\begin{equation}
 \langle X\rangle_{\ell,r_\ell}=x_0,
 \qquad
 \operatorname{var}_{\ell,r_\ell}(X)
 =\frac{x_0^2}{\beta_\ell}\longrightarrow0.
 \label{eq:continuum_X_concentration}
\end{equation}
Consequently, with
\begin{equation}
 \lambda(x_0):=\kappa\pi G(x_0),
\end{equation}
the normalized squeezed states satisfy
\begin{equation}
 \boxed{
 \|[\widehat{\mathcal A}_{\rm HS}-\lambda(x_0)]
 |\Psi_{\ell m}(r_\ell(x_0))\rangle\|
 \longrightarrow0,
 \qquad
 |\Psi_{\ell m}(r_\ell(x_0))\rangle\rightharpoonup0.
 }
 \label{eq:full_continuum_Weyl}
\end{equation}
As $x_0$ ranges over $(0,\infty)$, the values $\lambda(x_0)$ cover
$(\nu_{\min},\infty)$.  Thus one squeezed family gives constructive Fock-space
access to every interior point of the essential band of the full operator.
\end{theorem}

\begin{proof}
Equation~\eqref{eq:squeezed_X_moments} gives
\eqref{eq:continuum_X_concentration}.  Hence the spectral measure of
$\widehat X$ in the squeezed state concentrates at $x_0$.  The growth bound
$G(x)^2\le C(1+x)$ also shows that every squeezed state lies in the natural
domain of $G(\widehat X)$.  To pass from concentration of $X$ to norm
convergence of $G(X)$, fix a large $R$.  On $[0,R]$, continuity of $G$ is
uniform and Eq.~\eqref{eq:continuum_X_concentration} gives convergence in
probability to $x_0$; on $(R,\infty)$, the bound on $G^2$ together with the
uniformly bounded second moments of the associated Gamma laws makes the tail
uniformly small.  Therefore
\begin{equation}
 \|[\widehat{\mathcal N}_\ell-\kappa\pi G(x_0)]
 |\Psi_{\ell m}(r_\ell(x_0))\rangle\|
 \longrightarrow0.
 \label{eq:continuum_loss_Weyl}
\end{equation}
By Eq.~\eqref{eq:Cldecay},
$\|\widehat{\mathcal C}_\ell\|\to0$, so the same norm convergence holds for
the complete collision block.  Choosing one magnetic copy in each $\ell$
sector, for example $m=0$, places the states in mutually orthogonal angular
sectors; hence they converge weakly to zero in the full Fock space.  This is exactly the criterion in
Eq.~\eqref{eq:Weyl_criterion_reminder} for the full operator.  Finally, continuity of $G$, its global minimum $G(0)$, and $G(x)\to\infty$ imply that
$G((0,\infty))$ contains $(G(0),\infty)$, which is exactly the interior-band
coverage asserted above.
\end{proof}

Two distinctions are useful.  The theorem constructs interior-band Weyl
sequences by letting the angular sector vary, whereas
Theorem~\ref{thm:threshold_fixed_l} below reaches the endpoint
$\nu_{\min}$ within every fixed angular sector.  The common essential band of
each fixed block was already established abstractly in
Theorem~\ref{thm:common_essential}; the value of the squeezed construction is
that it exhibits explicit normalized Fock packets that realize the continuum
rather than merely identifying its spectral range.

The radial location of these continuum packets is equally informative.  From
Eq.~\eqref{eq:squeezed_j_moments} and
Eq.~\eqref{eq:continuum_squeeze_choice},
\begin{equation}
 \boxed{
 \langle j\rangle_{\ell,r_\ell(x_0)}
 =\frac{\beta_\ell^2}{4x_0}
 -\frac{\beta_\ell}{2}
 +\frac{x_0}{4}
 \sim\frac{\beta_\ell^2}{4x_0}.
 }
 \label{eq:continuum_j_corridor}
\end{equation}
Thus a fixed continuum energy is represented asymptotically along a quadratic
corridor in the $(\ell,j)$ plane.  This is the exact squeezed-state
counterpart of the semiclassical relation
$x_-(j,\ell)\sim\beta_\ell^2/(4j)$ derived later.  In this sense the
lower-left finite-Sonine corner, the continuum-energy corridors, and the
threshold escape are not different representations: they are different
regions of the same untruncated Fock tower.

\subsection{Threshold escape at large angular momentum}

The full-band theorem keeps $x_0>0$ fixed while $\ell$ grows.  To reach the
endpoint itself one lets the target energy drift to zero.  This gives the
complementary high-$\ell$ escape mechanism and makes precise how far beyond a
fixed Sonine window the packet must move.  In the Burnett realization the loss
is multiplication by $\kappa\pi G(x)$.  The small- and
large-$x$ expansions \eqref{eq:Gsmall}--\eqref{eq:Glarge} imply that there is
a constant $C_G$ such that
\begin{equation}
 0\le G(x)-G(0)\le C_G\sqrt{x},
 \qquad x\ge0.
 \label{eq:G_threshold_bound}
\end{equation}
Consequently the normalized squeezed state satisfies
\begin{equation}
 0\le
 \langle\Psi_{\ell m}(r)|\widehat{\mathcal N}_\ell|\Psi_{\ell m}(r)\rangle
 -\nu_{\min}
 \le
 \kappa\pi C_G\sqrt{\beta_\ell\e^{-2r}}.
 \label{eq:trial_loss_bound}
\end{equation}

\begin{theorem}[Constructive high-$\ell$ threshold limit]
Let $r_\ell\ge0$ satisfy
\begin{equation}
 \beta_\ell\e^{-2r_\ell}\longrightarrow0.
 \label{eq:r_escape_condition}
\end{equation}
Equivalently, with the parameter $t_\ell$ of
Eq.~\eqref{eq:t_r_q_relation},
\begin{equation}
 \frac{t_\ell}{\beta_\ell}\longrightarrow\infty.
 \label{eq:t_escape_condition}
\end{equation}
Then
\begin{equation}
 \langle\Psi_{\ell m}(r_\ell)|
 \widehat{\mathcal A}_\ell|\Psi_{\ell m}(r_\ell)\rangle
 \longrightarrow\nu_{\min}.
 \label{eq:trial_full_threshold}
\end{equation}
In particular,
\begin{equation}
 \boxed{
 \inf\spec(\widehat{\mathcal A}_\ell)\longrightarrow\nu_{\min}
 \qquad(\ell\to\infty).
 }
 \label{eq:bottom_limit}
\end{equation}
\end{theorem}

\begin{proof}
Equation~\eqref{eq:trial_loss_bound} and
$\beta_\ell\e^{-2r_\ell}\to0$ give convergence of the loss expectation to
$\nu_{\min}$.  Since
\begin{equation}
 |\langle\Psi_{\ell m}(r_\ell)|\widehat{\mathcal C}_\ell|
 \Psi_{\ell m}(r_\ell)\rangle|
 \le\|\widehat{\mathcal C}_\ell\|
 \longrightarrow0
\end{equation}
by Eq.~\eqref{eq:Cldecay}, the same limit holds for the full operator.  The
variational principle \cite{ReedSimon1980} gives the upper bound on the spectral bottom, while
$\widehat{\mathcal A}_\ell\ge
(\nu_{\min}-\|\widehat{\mathcal C}_\ell\|)I$ gives the matching
lower bound.
\end{proof}

This proof is deliberately stronger in interpretation than the soft argument
in Sec.~\ref{sec:global_spectral}.  It does not merely know that the threshold belongs to the
spectrum; it displays a normalized path through the exact radial modules that
finds it.  The fixed-$j$ and threshold-seeking paths are therefore
complementary: $|j_0,\ell,m\rangle$ with $j_0$ fixed remains on the
high-angular boundary where the loss grows as $\sqrt\ell$, whereas
$|\Psi_{\ell m}(r_\ell)\rangle$ has definite $\ell$ but its radial
probability distribution escapes to increasing $j$.

\subsection{How far in radial index must the threshold-seeking state move?}

The fixed-energy result in Sec.~\ref{sec:fock_essential_continuum} already
shows that an $O(1)$ radial energy $x_0>0$ is represented along the quadratic
corridor $\langle j\rangle\sim\beta_\ell^2/(4x_0)$.  For the threshold
path itself, put
\begin{equation}
 x_\ell:=\langle X\rangle_{\ell,r_\ell}
 =\beta_\ell\e^{-2r_\ell}.
\end{equation}
Combining the exact identities in Eqs.~\eqref{eq:squeezed_j_moments} and
\eqref{eq:squeezed_X_moments} gives
\begin{equation}
 \langle j\rangle_{\ell,r_\ell}
 =\frac{\beta_\ell^2}{4x_\ell}
 -\frac{\beta_\ell}{2}
 +\frac{x_\ell}{4}.
 \label{eq:threshold_j_exact}
\end{equation}
Hence the threshold condition $x_\ell\to0$ forces
\begin{equation}
 \frac{\langle j\rangle_{\ell,r_\ell}}{\beta_\ell^2}
 \longrightarrow\infty.
 \label{eq:superquadratic_threshold_escape}
\end{equation}
Equivalently, along large squeezes,
\begin{equation}
 \langle j\rangle_{\ell,r}
 \langle X\rangle_{\ell,r}
 \sim\frac{\beta_\ell^2}{4}.
 \label{eq:jx_product}
\end{equation}
For this explicit family, therefore, quadratic radial depth reaches any fixed
finite continuum energy, while convergence all the way to $x=0$ requires a
prefactor that itself diverges.  This is a constructive realization, not a
universal minimality theorem for all possible states; the same quadratic
geometry will reappear independently in the joint semiclassical symbol.

We can now state the noncommutation of limits without importing its second
line from the earlier essential-spectrum theorem:
\begin{equation}
 \boxed{
 \begin{aligned}
 N\ \text{fixed},\ \ell\to\infty:
 &\qquad
 \lambda^{(N)}_{\ell,k}\sim\kappa\pi\sqrt{2\ell},
 \\
 \text{squeezed escaping radial states},\ \ell\to\infty:
 &\qquad
 \inf\spec(\widehat{\mathcal A}_\ell)\to\nu_{\min}.
 \end{aligned}
 }
 \label{eq:noncommute}
\end{equation}
The first line is precisely the fixed-$N$ consequence of the relative
diagonalization \eqref{eq:large_l_relative_diagonalization}: every eigenvalue of
a fixed $(N+1)$-state Sonine compression shares the leading
$\kappa\pi\sqrt{2\ell}$ scale.  Equation~\eqref{eq:Ritz_large_l} resolves the
first $O(1)$ correction as the finite Hermite recurrence spectrum
$h_{N+1,k}$; throughout this statement $N$ is held fixed before
$\ell\to\infty$.  The second line follows a different family in which the radial
support itself escapes with $\ell$.
No fixed Sonine depth can therefore resolve the high-$\ell$ threshold region.
The obstruction is not a failure of low-order matrix elements: it is the
migration of the relevant state out of every fixed radial window.  The
squeezed ket makes this migration explicit and shows simultaneously how an
outward motion in the pair-created index can compensate the angular outward
displacement of the physical radial weight.

\section{Large radial index at fixed angular momentum}
\label{sec:large_j}

Large angular momentum at fixed radial depth suppresses relative mixing and
produces asymptotic diagonalization.  The complementary boundary of the index
plane behaves in the opposite way.  Far out along a \emph{fixed} angular
tower the diagonal and off-diagonal loss matrix elements grow on the same
$\sqrt j$ scale, so radial mixing does not collapse.  Instead the coefficients
become locally translation invariant and organize themselves into a universal
Toeplitz operator.

We now keep $\ell$ fixed and let the radial index tend to infinity.  This is a
different semiclassical regime of the same exact Jacobi operator.  Besides
describing the high-energy tail, we shall see that this regime already
contains a constructive mechanism for the threshold $x=0$.

To describe a fixed window around a large radial index, write the nearby sites
as $j+r$ with $r\in\mathbb Z$ fixed while $j\to\infty$.  In this moving
window the lower boundary $j=0$ recedes indefinitely far away, so the radial
half-lattice $\mathbb N_0$ becomes locally the translation-invariant lattice
$\mathbb Z$.  Let $S$ and $S^\ast$ denote the two bilateral shifts on
$\ell^2(\mathbb Z)$,
\begin{equation}
 (Su)_r=u_{r-1},
 \qquad
 (S^\ast u)_r=u_{r+1}.
 \label{eq:bilateral_shifts}
\end{equation}
Around the large index $j$, the rescaled Jacobi coefficients satisfy
\begin{equation}
 \frac{(J_\ell)_{j+r,j+r}}{j}\to2,
 \qquad
 \frac{(J_\ell)_{j+r,j+r+1}}{j}\to1
\end{equation}
for every fixed integer $r$.  Hence the local limiting Jacobi operator is
\begin{equation}
 J_\infty=2I+S+S^\ast,
 \label{eq:Jinf}
\end{equation}
or, equivalently,
\begin{equation}
 (J_\infty u)_r=2u_r+u_{r-1}+u_{r+1}.
\end{equation}
Because this operator is translation invariant, it is diagonalized by the
generalized lattice Fourier modes
\begin{equation}
 u_r(\theta)=\e^{ir\theta},
 \qquad -\pi\le\theta\le\pi.
\end{equation}
Indeed,
\begin{equation}
 Su(\theta)=\e^{-i\theta}u(\theta),
 \qquad
 S^\ast u(\theta)=\e^{i\theta}u(\theta),
\end{equation}
so
\begin{equation}
 J_\infty u(\theta)
 =\left(2+\e^{-i\theta}+\e^{i\theta}\right)u(\theta).
\end{equation}
The Fourier symbol is therefore
\begin{equation}
 q(\theta)=2+2\cos\theta
 =4\cos^2(\theta/2).
 \label{eq:qtheta}
\end{equation}
Here $\theta$ is not a physical velocity-space angle.  It is the Fourier phase,
or quasi-momentum, conjugate to translation along the locally homogeneous
radial-index lattice.  In particular, $\theta=\pi$ corresponds to the
alternating lattice pattern $u_r=(-1)^r$.
Using \eqref{eq:Glarge},
\begin{equation}
 \frac{G(jq)}{\sqrt j}
 \longrightarrow
 \sqrt{2q}.
\end{equation}
Hence the limiting loss symbol is
\begin{equation}
 g_\infty(\theta)
 =
 \sqrt{2q(\theta)}
 =
 2\sqrt2\,|\cos(\theta/2)|.
 \label{eq:ginf}
\end{equation}

\begin{theorem}[Local high-$j$ Toeplitz limit]
Fix $\ell$ and integers $r,s$.  Then
\begin{equation}
 \boxed{
 \lim_{j\to\infty}
 \frac{
 \mathcal A^{(\ell)}_{j+r,j+s}
 }{\kappa\pi\sqrt j}
 =
 \tau_{s-r},
 }
 \label{eq:Toeplitz_limit}
\end{equation}
where
\begin{equation}
 \tau_d
 =
 \frac1{2\pi}
 \int_{-\pi}^{\pi}
 2\sqrt2\cos(\theta/2)\,
 \e^{-id\theta}\dd\theta
 =
 \frac{4\sqrt2}{\pi}
 \frac{(-1)^d}{1-4d^2}
 .
 \label{eq:tau}
\end{equation}
The limit is independent of $\ell$.
\end{theorem}

\begin{proof}
Shift the radial basis so that $|j+r\rangle$ is identified with the fixed
site $|r\rangle$ in $\ell^2(\mathbb Z)$.  On every fixed finite window,
$J_\ell/j$ converges coefficientwise to $J_\infty$ in
\eqref{eq:Jinf}.  Hence every fixed polynomial matrix element converges to the
corresponding matrix element of $J_\infty$.

To pass from polynomials to the collision-frequency function, set
\begin{equation}
 f_j(x):=\frac{G(jx)}{\sqrt j}.
\end{equation}
The explicit formula for $G$ implies that
$G(y)-\sqrt{2y}$ is bounded on $[0,\infty)$; consequently
\begin{equation}
 \sup_{x\ge0}|f_j(x)-\sqrt{2x}|=O(j^{-1/2}).
 \label{eq:fj_uniform}
\end{equation}
It remains to pass from polynomial matrix elements of $J_\ell/j$ to those of
its square root.

Let $u=\sum_r u_r e_r$ be any fixed finitely supported vector on the moving
lattice, and for all sufficiently large $j$ let
$u^{(j)}:=\sum_r u_r|j+r\rangle$ be its translated vector in the original
radial half-lattice.  Denote by $\mu_j^u$ the positive spectral measure of
$J_\ell/j$ at $u^{(j)}$, and by $\mu_\infty^u$ the spectral measure of
$J_\infty$ at $u$.  For every polynomial $p$,
\begin{equation}
 \langle u^{(j)},p(J_\ell/j)u^{(j)}\rangle
 \longrightarrow
 \langle u,p(J_\infty)u\rangle,
 \label{eq:polynomial_quadratic_convergence}
\end{equation}
because a polynomial of fixed degree probes only a fixed finite window of
Jacobi coefficients.  Hence all moments of $\mu_j^u$ converge to those of
$\mu_\infty^u$.  The first moment bound makes the family tight, since its
mass above $R$ is bounded by $R^{-1}\int x\,\dd\mu_j^u$.  Along any weakly
convergent subsequence, the $n$th moment also passes to the limit: for $R>0$,
\begin{equation}
 \int_R^\infty x^n\,\dd\mu_j^u(x)
 \le R^{-n}\int_R^\infty x^{2n}\,\dd\mu_j^u(x),
\end{equation}
and the $2n$th moments are uniformly bounded.  Thus the limit has the same
moments as $\mu_\infty^u$.  Moreover, the common moments force any such weak limit to be supported on
$[0,4]$.  Indeed, if $\mu$ is a subsequential limit, then for every $R>4$ and
integer $n\ge1$,
\begin{equation}
 \mu([R,\infty))
 \le R^{-n}\int x^n\,\dd\mu(x)
 =R^{-n}\int x^n\,\dd\mu_\infty^u(x)
 \le \|u\|^2(4/R)^n,
\end{equation}
so letting $n\to\infty$ gives $\mu([R,\infty))=0$.  On the common compact
support $[0,4]$, equality of all moments implies equality of the measures,
because polynomials are dense in $C([0,4])$ by the Weierstrass approximation
theorem \cite{Rudin1976}.  Thus
\begin{equation}
 \mu_j^u\rightharpoonup\mu_\infty^u.
\end{equation}

The remaining point is that the limiting function $\sqrt{x}$ is unbounded on
$[0,\infty)$, so weak convergence of spectral measures alone is not enough.
Its spectral tail is nevertheless uniformly controlled by the first moment.
Indeed,
\begin{equation}
 \langle u^{(j)},(J_\ell/j)u^{(j)}\rangle=O(1),
\end{equation}
and therefore, for $R>0$,
\begin{equation}
 \int_R^\infty\sqrt{x}\,\dd\mu_j^u(x)
 \le
 \frac1{\sqrt R}\int_R^\infty x\,\dd\mu_j^u(x)
 \le\frac{C_u}{\sqrt R}.
 \label{eq:sqrt_uniform_integrability}
\end{equation}
Apply weak convergence first to a bounded continuous truncation of
$\sqrt{x}$ and then let $R\to\infty$ using
\eqref{eq:sqrt_uniform_integrability}.  This gives
\begin{equation}
 \langle u^{(j)},\sqrt{2J_\ell/j}\,u^{(j)}\rangle
 \longrightarrow
 \langle u,\sqrt{2J_\infty}\,u\rangle.
 \label{eq:sqrt_quadratic_convergence}
\end{equation}
Only at this point do we use polarization.  Applying
\eqref{eq:sqrt_quadratic_convergence} to
$u=e_r\pm e_s$ and $u=e_r\pm i e_s$ reconstructs the off-diagonal matrix
element and yields
\begin{equation}
 \langle j+r|\sqrt{2J_\ell/j}|j+s\rangle
 \longrightarrow
 \langle r|\sqrt{2J_\infty}|s\rangle.
\end{equation}
Together with the uniform functional-calculus estimate
\eqref{eq:fj_uniform}, this proves
\begin{equation}
 \langle j+r|f_j(J_\ell/j)|j+s\rangle
 \longrightarrow
 \langle r|\sqrt{2J_\infty}|s\rangle.
\end{equation}

Fourier transformation diagonalizes $J_\infty$ with symbol
\eqref{eq:qtheta}, giving the Fourier coefficient
\eqref{eq:tau}.  The gain contribution does not affect the normalized limit:
$\widehat{\mathcal C}_\ell$ is bounded, and therefore
\begin{equation}
 \frac{
 \langle j+r|\widehat{\mathcal C}_\ell|j+s\rangle
 }{\sqrt j}
 \to0.
\end{equation}
\end{proof}
The technical length of the proof is concentrated in one point: local
coefficient convergence of $J_\ell/j$ immediately controls polynomials, but the
hard-sphere asymptotic function is the unbounded square root.  The uniform
estimate \eqref{eq:fj_uniform} removes the model-specific remainder in $G$, and
the spectral-measure/tail argument justifies passage from polynomials to
$\sqrt{x}$.  Once that step is established, polarization and Fourier analysis
produce the Toeplitz coefficients directly.

The first coefficients are
\begin{equation}
 \tau_0=\frac{4\sqrt2}{\pi},
 \qquad
 \tau_1=\frac{4\sqrt2}{3\pi},
 \qquad
 \tau_2=-\frac{4\sqrt2}{15\pi},
 \qquad
 \tau_3=\frac{4\sqrt2}{35\pi}.
\end{equation}
Thus, for example,
\begin{equation}
 \mathcal A^{(\ell)}_{jj}
 \sim
 4\kappa\sqrt{2j},
 \qquad
 \mathcal A^{(\ell)}_{j,j+1}
 \sim
 \frac{4\kappa}{3}\sqrt{2j}.
 \label{eq:largej_first}
\end{equation}
Unlike the large-$\ell$ fixed-$j$ limit, the large-$j$ tail does not become
diagonal or even nearest-neighbor.  It approaches a universal Toeplitz matrix
whose off-diagonal coefficients decay only algebraically as $d^{-2}$.  The
loss has therefore changed from an angularly induced near-diagonal regime to a
radially delocalized one.

\subsection{The threshold as the large-squeeze limit}

The Toeplitz symbol contains an important clue:
\begin{equation}
 g_\infty(\pi)=0.
 \label{eq:toeplitz_zero}
\end{equation}
Because the theorem concerns the matrix $\mathcal A^{(\ell)}/\sqrt j$,
this zero should not
be read as a zero eigenvalue of the unscaled collision operator.  It says that
coherent combinations of neighboring high-$j$ states can cancel the leading
$O(\sqrt j)$ loss scale.  The squeezed sequence constructed in
Sec.~\ref{sec:squeezed_radial_path} shows that this cancellation is the local
high-$j$ signature of the actual threshold Weyl sequence.

Fix $\ell$ and let $r\to\infty$ in
$|\Psi_{\ell m}(r)\rangle$.  From Eq.~\eqref{eq:squeezed_X_moments},
\begin{equation}
 \langle X\rangle_{\ell,r}=\beta_\ell\e^{-2r}\longrightarrow0.
\end{equation}
Using Eq.~\eqref{eq:G_threshold_bound},
\begin{equation}
 \|\bigl(\widehat{\mathcal N}_\ell-\nu_{\min}\bigr)
 |\Psi_{\ell m}(r)\rangle\|^2
 \le
 (\kappa\pi C_G)^2\beta_\ell\e^{-2r}
 \longrightarrow0.
 \label{eq:weyl_loss_norm}
\end{equation}

The ket expansion \eqref{eq:squeezed_ket_expansion} also makes the escape in
radial index explicit.  For every fixed $j$,
\begin{equation}
 \langle j,\ell,m|\Psi_{\ell m}(r)\rangle
 =(1-q^2)^{\beta_\ell/2}
 (-q)^j\sqrt{\frac{(\beta_\ell)_j}{j!}}
 \longrightarrow0
 \qquad(r\to\infty),
\end{equation}
because $q=\tanh r\to1$ while the normalization prefactor tends to zero.
Since the vectors remain normalized and the finite span of the radial basis is
dense,
\begin{equation}
 |\Psi_{\ell m}(r)\rangle\rightharpoonup0
 \qquad(r\to\infty).
 \label{eq:weyl_weak}
\end{equation}
By the same weak-to-strong property of compact operators used in the proof of
Eq.~\eqref{eq:Cldecay}, compactness of the fixed-$\ell$ gain block implies
\begin{equation}
 \|\widehat{\mathcal C}_\ell|\Psi_{\ell m}(r)\rangle\|
 \longrightarrow0.
 \label{eq:weyl_gain}
\end{equation}
We obtain the following constructive form of the threshold statement.

\begin{theorem}[Threshold Weyl sequence in every angular sector]\label{thm:threshold_fixed_l}
For every fixed $\ell$ and any magnetic copy $m$,
\begin{equation}
 \boxed{
 \|\bigl(\widehat{\mathcal A}_\ell-\nu_{\min}\bigr)
 |\Psi_{\ell m}(r)\rangle\|
 \longrightarrow0,
 \qquad
 |\Psi_{\ell m}(r)\rangle\rightharpoonup0,
 \qquad r\to\infty.
 }
 \label{eq:threshold_Weyl_sequence}
\end{equation}
Hence $\nu_{\min}\in\specess(\widehat{\mathcal A}_\ell)$, equivalently
$\nu_{\min}\in\specess(\mathfrak A_\ell)$, by the same Weyl criterion
\eqref{eq:Weyl_criterion_reminder}.
\end{theorem}

The connection with the Toeplitz phase is now explicit rather than
suggestive.  In Eq.~\eqref{eq:squeezed_ket_expansion} the coefficient ratio
between neighboring high-$j$ components approaches
\begin{equation}
 \frac{c_{j+1}}{c_j}
 =-q\sqrt{\frac{j+\beta_\ell}{j+1}}
 \longrightarrow-1
\end{equation}
when first $j$ is large at fixed $\ell$ and then $r$ is large.  Thus the broad
radial packet carries asymptotically the alternating phase $(-1)^j$, namely
the lattice quasi-momentum $\theta=\pi$ at which the Toeplitz symbol
vanishes.  The large-$j$ Toeplitz zero therefore encodes the cancellation of
the leading $\sqrt j$ loss scale, while the exact squeezed sequence identifies
the surviving unscaled value as the physical threshold $\nu_{\min}$.

\section{Semiclassical synthesis in the \texorpdfstring{$(\ell,j)$}{(l,j)} plane}
\label{sec:semiclassical_synthesis}

The preceding two sections approached the radial operator from two different
edges of its index set.  At fixed radial depth and large angular momentum, the
operator becomes nearly diagonal at leading order and the first correction is
a universal Hermite--Jacobi coupling.  At fixed angular momentum and large
radial index, by contrast, the local tail approaches a nonlocal Toeplitz
operator.  A natural question is therefore what these two limits are telling
us about the operator away from either edge, when both $j$ and $\ell$ are
large.

The discussion in this section is a semiclassical synthesis of the exact
Jacobi structure, not an additional theorem.  Its purpose is to identify
variables and local symbols that interpolate between the two proved
asymptotic regimes and may guide a future joint-limit analysis.

The diagonal scale of the exact Jacobi operator is
\begin{equation}
 E_{j\ell}:=2j+\beta_\ell,
 \qquad
 \beta_\ell=\ell+\frac32.
 \label{eq:joint_energy_scale}
\end{equation}
Since the total oscillator degree is $N=2j+\ell$, one has
$E_{j\ell}=N+3/2$.  Thus $E_{j\ell}$ is the natural kinetic-energy scale
associated with the Burnett/Fock state $(j,\ell)$.  To distinguish how that
large degree is distributed between angular and radial excitation, introduce
\begin{equation}
 \sigma_{j\ell}
 :=\frac{\beta_\ell}{2j+\beta_\ell},
 \qquad 0<\sigma_{j\ell}\le1.
 \label{eq:sigma_joint}
\end{equation}
The limit $j\to\infty$ at fixed $\ell$ corresponds to $\sigma\to0$, whereas
$\ell\to\infty$ at fixed $j$ corresponds to $\sigma\to1$.

Around a large radial index $j$, freeze the slowly varying Jacobi coefficients
on a finite local window.  The resulting local symbol of $J_\ell$ is
\begin{equation}
 x_{j\ell}(\theta)
 \approx
 2j+\beta_\ell
 +2\sqrt{j(j+\beta_\ell)}\cos\theta,
 \qquad -\pi\le\theta\le\pi.
 \label{eq:joint_jacobi_symbol}
\end{equation}
Equivalently,
\begin{equation}
 \boxed{
 x_{j\ell}(\theta)
 \approx
 E_{j\ell}
 \left[
 1+\sqrt{1-\sigma_{j\ell}^{\,2}}\cos\theta
 \right].
 }
 \label{eq:joint_symbol_sigma}
\end{equation}
This formula provides a simple interpolation between the two boundary
regimes.  When $\sigma\to0$, the relative bandwidth remains of order one and
\eqref{eq:joint_symbol_sigma} reduces, after division by $j$, to the symbol
$4\cos^2(\theta/2)$ that generated the large-$j$ Toeplitz limit.  When
$\sigma\to1$, the relative bandwidth collapses and the leading operator is
nearly diagonal; resolving the next order produces the large-$\ell$
Hermite--Jacobi limit.

The local energy interval associated with
\eqref{eq:joint_jacobi_symbol} has edges
\begin{equation}
 \boxed{
 x_\pm(j,\ell)
 =
 \left(\sqrt{j+\beta_\ell}\pm\sqrt j\right)^2.
 }
 \label{eq:xpm_joint}
\end{equation}
The lower edge $x_-$ is particularly informative because the essential band
is parameterized by the physical radial energy $x\ge0$, with the common
threshold attained at $x=0$.  In view of
Sec.~\ref{sec:fock_essential_continuum}, contours $x_-(j,\ell)\approx x_0$
may be read, at the present semiclassical level, as approximate locations of
Fock packets representing the continuum energy $\kappa\pi G(x_0)$.  At fixed
$j$, $x_-$ grows with $\ell$, which is another way to see why a fixed Sonine
depth cannot represent the high-$\ell$ threshold region.  Conversely, when
$j\gg\beta_\ell$,
\begin{equation}
 x_-(j,\ell)
 =
 \frac{\beta_\ell^2}{4j}
 \left[1+O\!\left(\frac{\beta_\ell}{j}\right)\right].
 \label{eq:xminus_mixed_scaling}
\end{equation}
Thus keeping the lower local energy edge of order one as $\ell\to\infty$
suggests the mixed scaling $j=O(\ell^2)$.  The exact squeezed family gives an
independent check: Eq.~\eqref{eq:continuum_j_corridor} places a prescribed
energy $x_0>0$ at
$\langle j\rangle\sim\beta_\ell^2/(4x_0)$, while
Eq.~\eqref{eq:threshold_j_exact} shows that $x_0\downarrow0$ pushes the
required radial depth beyond every fixed multiple of $\ell^2$.  The
local-symbol edge and the exact squeezed family therefore identify the same
quadratic continuum corridors, with the threshold as their limiting escape.

\paragraph{Three qualitative regimes.}
The two proved boundary asymptotics, the explicit threshold families, and the
joint symbol now admit a unified spectral interpretation.

First, at fixed $\ell$ and $j\to\infty$ the loss matrix elements grow like
$\sqrt j$, but radial mixing remains of the same order and converges locally to
a translation-invariant Toeplitz pattern.  The symbol
$2\sqrt2|\cos(\theta/2)|$ vanishes at $\theta=\pi$.  The exact Weyl family
shows what this means in the unscaled problem: an alternating high-$j$ phase
cancels the leading divergent loss and leaves the finite threshold
$\nu_{\min}$.  This is a continuum-forming, delocalized mechanism rather than
a new localized discrete branch.

Second, at fixed radial depth and $\ell\to\infty$ the situation is opposite.
The leading loss grows as $\sqrt\ell I$, nearest-neighbor mixing is only
$O(1)$, farther couplings vanish, and $\|\widehat{\mathcal C}_\ell\|\to0$.  Relative to
its leading scale the loss therefore collapses toward the identity.  No
coherent combination inside a fixed Sonine window can compensate this angular
outward shift; every such finite section is driven to large collision
frequencies.

Third, the exact high-$\ell$ threshold is recovered by moving into the
interior of the $(\ell,j)$ plane.  The trial family proves that a fixed
low-energy layer requires typical $j$ of quadratic order in $\ell$, while the
local edge $x_-$ gives the same scale semiclassically.  The threshold itself
requires still deeper radial escape.  The interior therefore interpolates
between angularly driven near-diagonality and radial Toeplitz cancellation.

The resulting picture is genuinely two-parameter.  Large $j$ at fixed $\ell$
is a high-radial-excitation direction whose local spectrum contains a zero
of its Toeplitz symbol at quasi-momentum $\theta=\pi$ after $\sqrt j$
scaling; large $\ell$ at fixed $j$ is a
high-angular-momentum direction whose radial weight is displaced outward and
whose relative bandwidth collapses.  The variables
$(E_{j\ell},\sigma_{j\ell})$ separate these effects and provide natural
coordinates for future joint asymptotic analysis.

\subsection{Semiclassical phase-plane diagnostics}

The joint symbol gives several scalar functions of $(\ell,j)$ that can be
visualized without introducing a fictitious scalar ``operator
$\mathcal A(j,\ell)$''.  The full collision operator has the three indices
$(\ell,j,j')$; the figures below instead display precisely defined diagnostics
of the frozen Jacobi symbol \eqref{eq:joint_jacobi_symbol}.  They therefore
summarize the semiclassical geometry used in the preceding discussion, while
making no replacement of the compact gain operator by a scalar model.

\begin{figure}[p]
\centering
\includegraphics[width=0.80\textwidth]{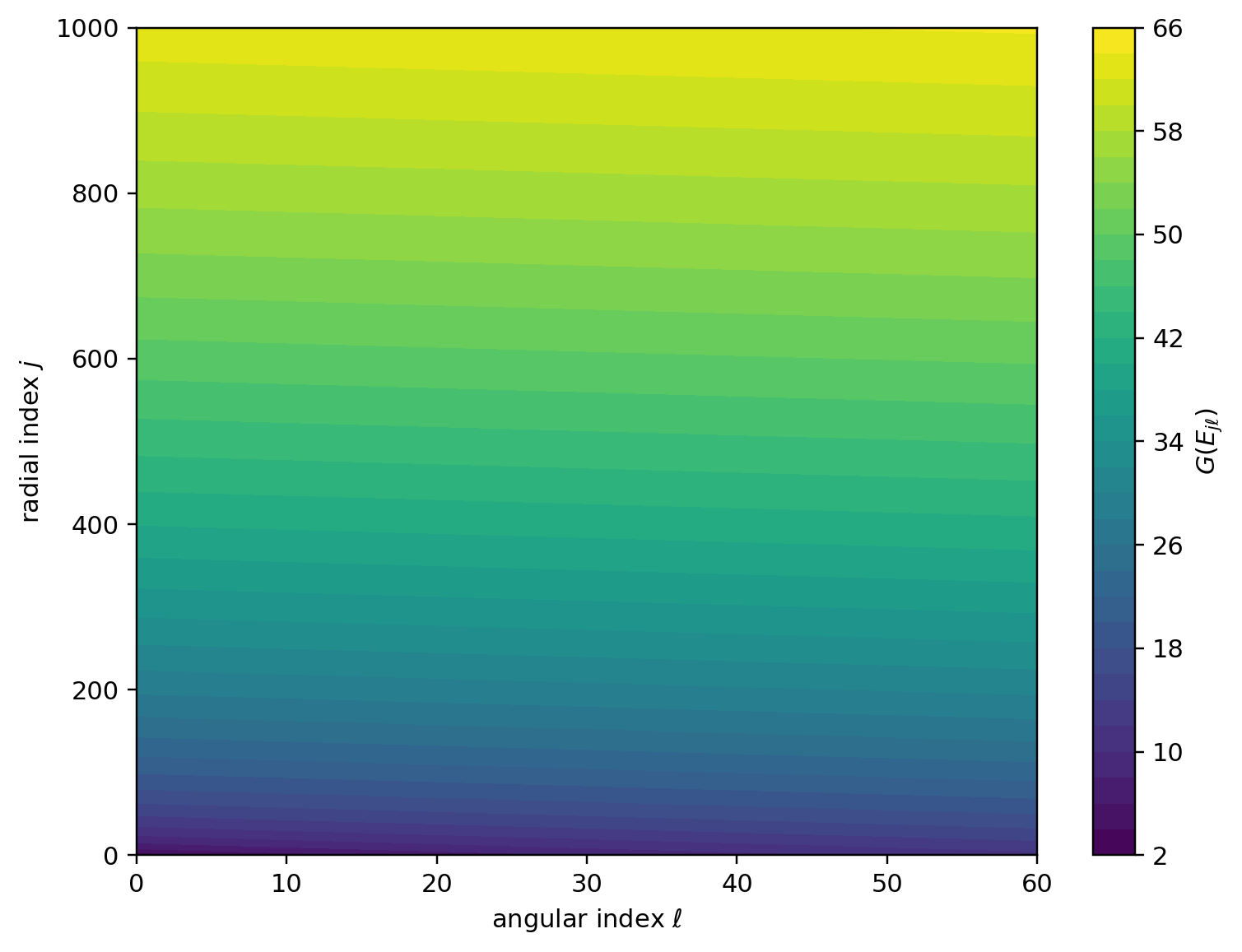}
\caption{Collision-frequency scale at the center of the frozen local band,
$G(E_{j\ell})$ with $E_{j\ell}=2j+\beta_\ell$.  The surface is a scalar
diagnostic of the exact loss function evaluated at the local Jacobi energy
center; it is not the diagonal matrix element of $G(J_\ell)$.}
\label{fig:phase_center_scale}
\end{figure}

\begin{figure}[p]
\centering
\includegraphics[width=0.80\textwidth]{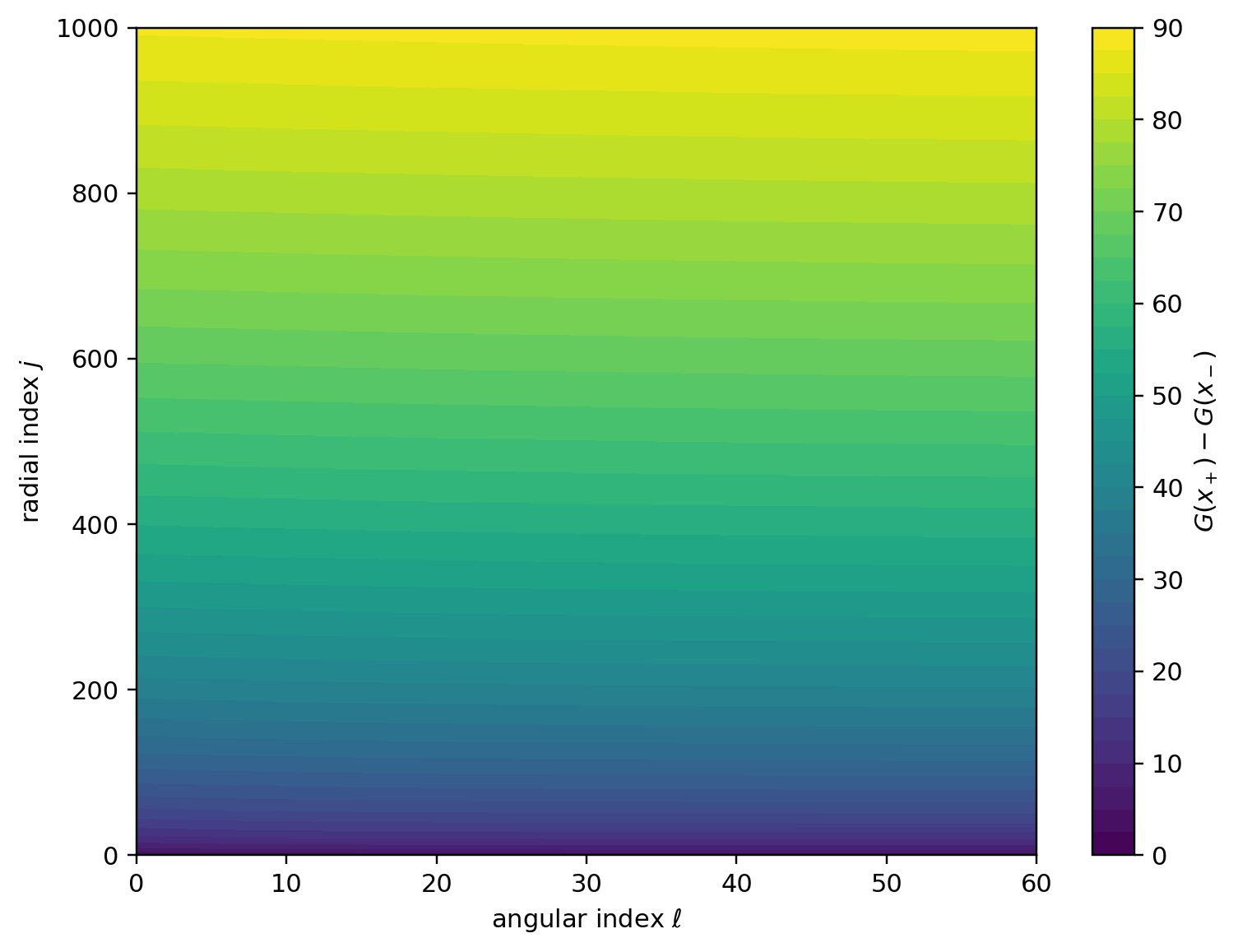}
\caption{Collision-frequency bandwidth of the frozen symbol,
$G(x_+)-G(x_-)$.  Broad bands characterize the radial/high-energy side of the
phase plane, while the bandwidth collapses toward the fixed-$j$, large-$\ell$
edge.}
\label{fig:phase_frequency_bandwidth}
\end{figure}

\begin{figure}[p]
\centering
\includegraphics[width=0.80\textwidth]{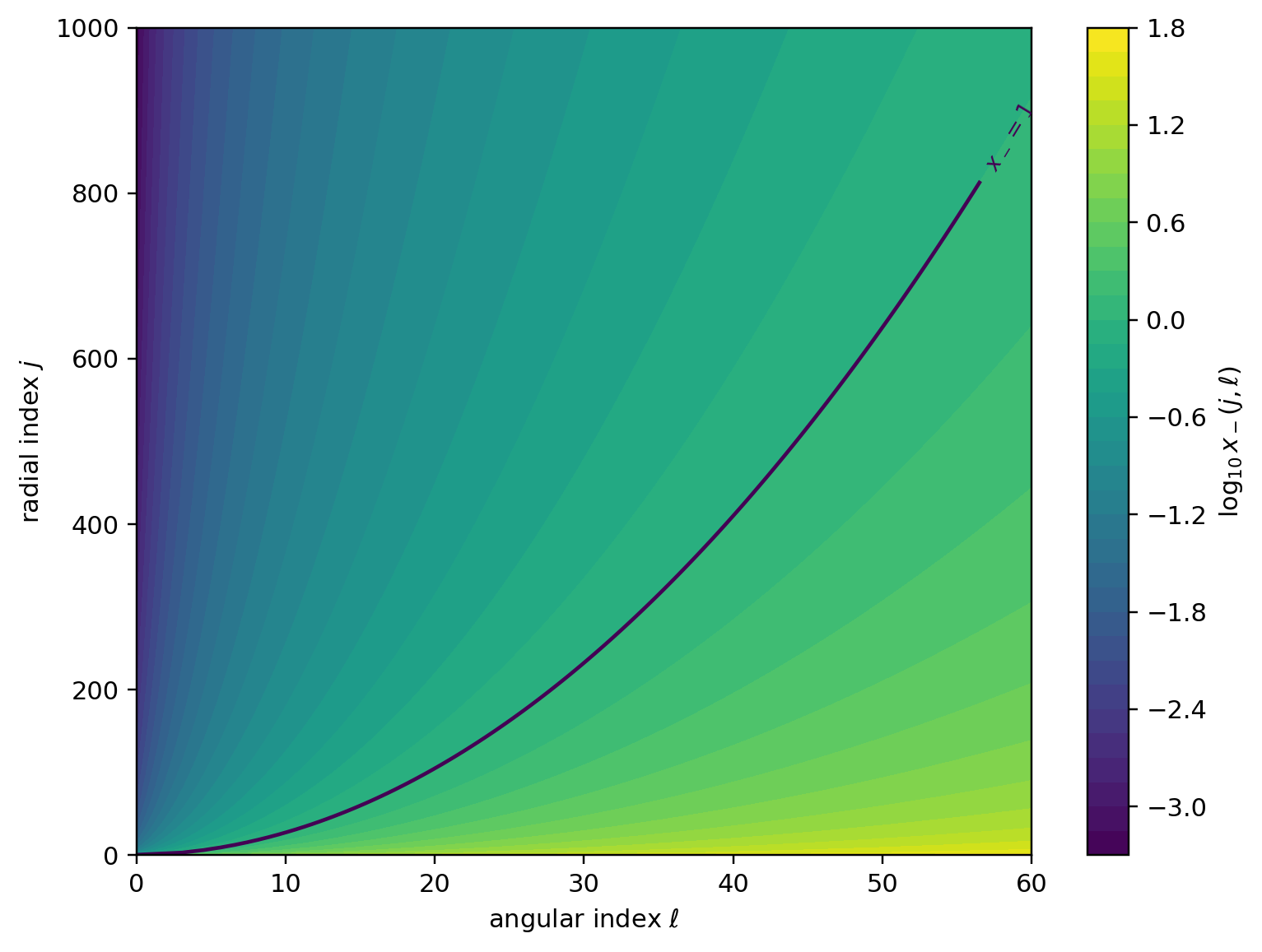}
\caption{Logarithm of the lower local energy edge,
$\log_{10}x_-(j,\ell)$.  The contour $x_-=1$ marks an order-one energy level
and bends toward $j\propto\ell^2$, illustrating the mixed scaling in
Eq.~\eqref{eq:xminus_mixed_scaling}.  The exact squeezed family realizes the same quadratic scale for fixed physical
energy through Eq.~\eqref{eq:continuum_j_corridor}; decreasing $x_0$ pushes
the corresponding continuum corridor progressively farther into the radial
tail.}
\label{fig:phase_xminus}
\end{figure}

\begin{figure}[p]
\centering
\includegraphics[width=0.80\textwidth]{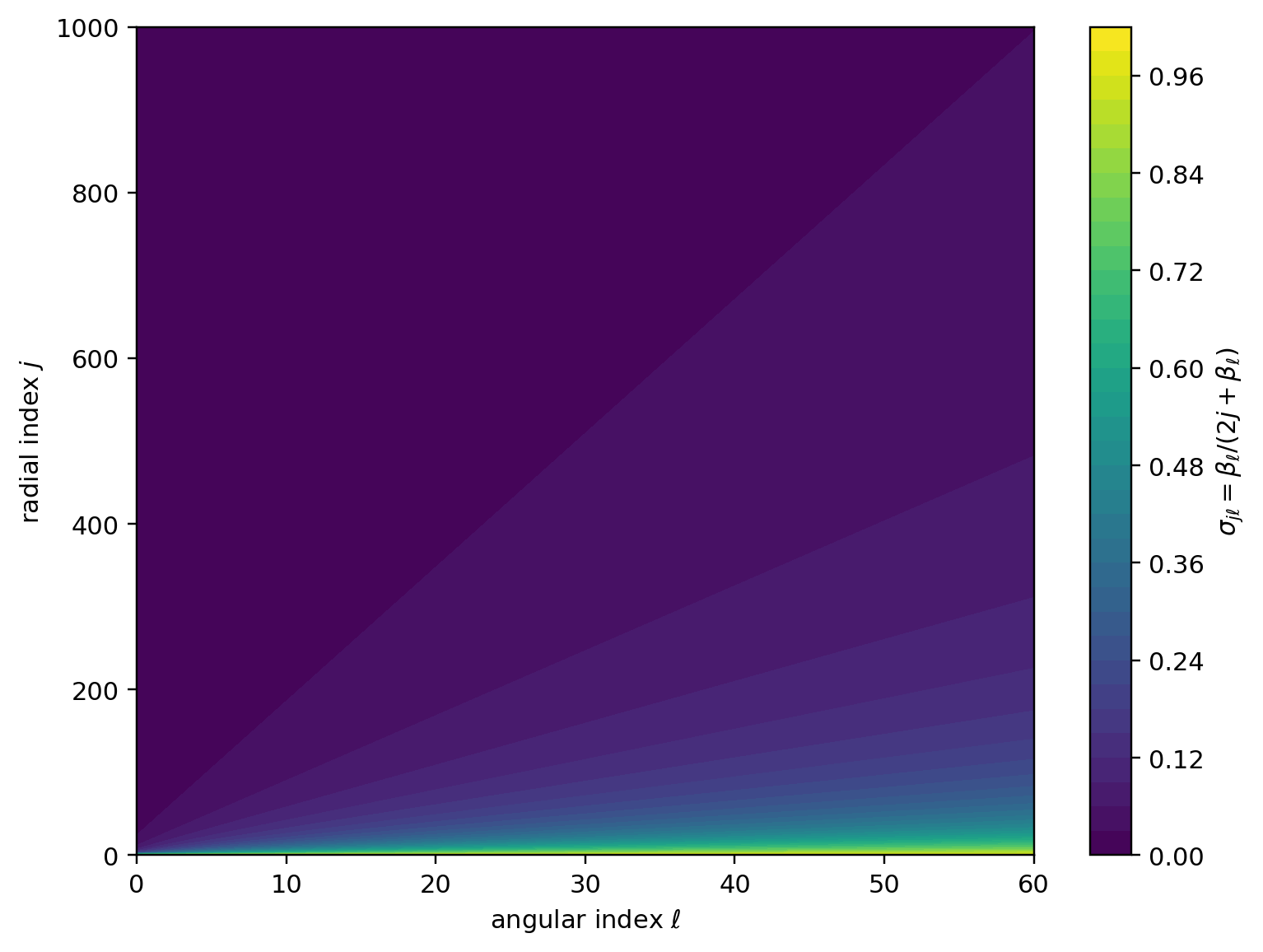}
\caption{Angular/radial balance parameter
$\sigma_{j\ell}=\beta_\ell/(2j+\beta_\ell)$.  The radial/high-energy boundary
has $\sigma\to0$, whereas the high-angular-momentum, fixed-depth boundary has
$\sigma\to1$.  Intermediate values parameterize the joint semiclassical
regime.}
\label{fig:phase_sigma}
\end{figure}

These diagnostics are not additional spectral theorems; their role is to
make the already derived joint symbol geometrically legible.  In particular,
Fig.~\ref{fig:phase_xminus} is now supported by the exact squeezed-state
scaling in Eqs.~\eqref{eq:continuum_j_corridor} and
\eqref{eq:threshold_j_exact}: the semiclassical lower edge and the explicit
radial construction identify the same continuum corridors and their limiting
threshold escape.

\clearpage
\section{Finite Sonine resolution versus global spectrum}
\label{sec:practical_regimes}

The global estimates raise an obvious interpretive question.  In the
$(\ell,j)$ picture, traditional finite Sonine calculations occupy a finite
lower-left window, whereas the continuum is represented by packets whose
coefficients can migrate arbitrarily far along the radial tower.  Yet for
ordinary zero-frequency transport the classical low-order success is
undeniable.  Equations~\eqref{eq:second_sonine_ratios} and
\eqref{eq:transport_infinite_benchmarks} show that the first Sonine values lie
only about $1.6\%$ below the infinite-order hard-sphere viscosity and about
$2.5\%$ below the thermal-conductivity benchmark; the two-state values reduce
those discrepancies to roughly $0.12\%$ and $0.24\%$.  Nothing in the present
global analysis contradicts that efficiency.

The distinction is between accurate approximation of selected low-order
responses and uniform resolution of the operator.  A fixed Sonine section can
converge rapidly for a matrix element such as
$\langle e|\mathfrak A_\ell^{-1}|e\rangle$ while still missing the Weyl
sequences and the mechanism producing the essential continuum, because those
states escape every fixed radial window.  The large-index results should
therefore be read as statements about what no fixed radial closure can contain,
not as a prescription to use many moments in ordinary continuum calculations.

When the collision operator is embedded in an inhomogeneous kinetic problem,
large radial or angular sectors can become relevant at kinetic spatial or
temporal scales, in Knudsen layers, and for strongly non-Maxwellian states
\cite{GradRarefied1949,ChapmanCowling1970}.  The present paper does not derive a
truncation-error estimate for those regimes.  It only identifies the global
radial structures that a uniformly valid representation would have to resolve.
In particular, large $j$ should not be equated naively with large Mach number:
if the reference Maxwellian follows the local velocity and temperature, a bulk
displacement or uniform temperature change need not require large polynomial
degree.

\section{Discussion}
\label{sec:discussion}

The exact radial construction, its low-order validation, and the joint
$(\ell,j)$ synthesis provide complementary views of the hard-sphere operator.  In that setting, the three Fock-space kinetic problems considered
in this series illustrate three different uses of the same representation
philosophy.  In the LFH model
\cite{KarlinLFH2026}, the oscillator grading organizes a differential
relaxation operator while intertwining controls the moving hydrodynamic
realization.  For nonlinear Maxwell molecules \cite{KarlinMaxwell2026},
lift--fusion intertwining produces an exactly graded bilinear collision
vertex.  The present hard-sphere problem is less algebraically soluble:
grading is lost, but exact radialization and the loss--gain functional
calculus remain.

For Maxwell molecules, Fock-space reformulation is rewarded by exact
oscillator grading.  Hard spheres remove that algebraic simplification:
the entire function $F$ in \eqref{eq:Fdef} generates infinitely many radial
degrees and the collision operator mixes the radial Fock index.

The present calculation shows that loss of grading does not make the Fock reformulation
empty.  The payoff changes from algebraic grading to analytic radial
structure.  The canonical oscillator representation contains a rotationally
adapted $\mathfrak{su}(1,1)$ pair algebra; after angular reduction, the whole
three-dimensional collision problem becomes one exact radial operator per
$\ell$.  Its noncompact part is an explicit function of the Jacobi generator,
and an exact unitary intertwiner relates the radial Bargmann and physical
Burnett coordinates.

The global analysis makes the elementary Boltzmann loss--gain structure the
central spectral guide.  The gain is the difficult piece to construct
exactly, because it contains the collision redistribution geometry, but it is
compact.  The loss is analytically simple in the physical radial coordinate,
yet it is noncompact and controls the whole essential continuum, including
its threshold, together with the high-$\ell$ relative diagonalization and the
high-$j$ Toeplitz tail.  The gain is not
irrelevant: it supplies the compact corrections responsible for quantitative
transport shifts and for isolated discrete levels below the loss continuum.
The technically difficult part and the globally noncompact part are therefore
not the same part of the Boltzmann operator.

This operator-first viewpoint clarifies the precise limitation of fixed Sonine
approximations.  They can be extremely accurate for low-order hydrodynamic
responses, as the transport sequences demonstrate, but no fixed radial depth
can represent the continuum-forming states.  The squeezed $\mathfrak{su}(1,1)$ family
makes this failure constructive rather than qualitative.  For any prescribed
$x_0>0$, normalized packets reach the essential-spectrum value
$\kappa\pi G(x_0)$ along the quadratic corridor
$j\sim\beta_\ell^2/(4x_0)$; convergence to the threshold itself requires still
deeper radial escape.  At fixed $\ell$, the same squeezed family becomes a
threshold Weyl sequence, and its asymptotic phase $(-1)^j$ is precisely the
$\theta=\pi$ zero of the universal Toeplitz tail.  The large-$\ell$, large-$j$,
and exact squeezed constructions therefore describe different views of one
continuum-forming mechanism.

Several extensions remain natural but are not required for the present
linear analysis.  A quantitative estimate
\begin{equation}
 \|\widehat{\mathcal C}_\ell\|\le C\ell^{-\alpha}
\end{equation}
would immediately sharpen \eqref{eq:bottom_bounds_soft} into a quantitative
approach rate to the essential threshold.  A Birman--Schwinger analysis of
the exact radial kernel may also provide a representation-level explanation
of the sharp critical angular momentum found by Klaus.  Finally, the same
pair algebra and angular decomposition are natural candidates for organizing
the nonlinear hard-sphere collision vertex, where two radial modules must be
coupled rather than one radial tower propagated.

The main conclusion is therefore both structural and practical.  The
Bargmann--Fock representation is useful for hard spheres not because it
supplies another polynomial basis, but because it retains the complete radial
operator before truncation.  That exact operator explains why low-order
Burnett/Sonine approximations can converge rapidly for hydrodynamic transport
while still missing the continuum-forming radial mechanisms required by
global spectral and kinetic-scale questions.

\appendix

\section{Special-function identities}

For convenience we collect the identities used in the radial realization;
standard special-function conventions follow \cite{NIST2010}.
The definition of ${}_0F_1$ and the Pochhammer symbol was given in
Eqs.~\eqref{eq:0F1_definition_main} and \eqref{eq:pochhammer_def}; we record
here the special-function identities used later.
Let $I_\nu$ denote the modified Bessel function of the first kind and
$J_\nu$ the Bessel function of the first kind.  Then the modified-Bessel
representation is
\begin{equation}
 {}_0F_1(; \nu+1;z)
 =
 \Gamma(\nu+1)
 z^{-\nu/2}I_\nu(2\sqrt z),
\end{equation}
while for negative argument
\begin{equation}
 {}_0F_1\!\left(;\nu+1;-\frac{x^2}{4}\right)
 =
 \Gamma(\nu+1)
 \left(\frac2x\right)^\nu
 J_\nu(x).
\end{equation}
In three dimensions
\begin{equation}
 \beta_\ell=\ell+\frac32,
 \qquad
 \nu=\ell+\frac12.
\end{equation}
For $\ell=0$,
\begin{equation}
 {}_0F_1(;3/2;z)
 =
 \frac{\sinh(2\sqrt z)}{2\sqrt z}.
\end{equation}

The differential equation
\begin{equation}
 z y''+b y'-y=0
\end{equation}
implies the radial lowering identity
\begin{equation}
 \left(
 \xi\partial_\xi^2+\beta_\ell\partial_\xi
 \right)
 {}_0F_1(; \beta_\ell;x\xi)
 =
 x\,{}_0F_1(; \beta_\ell;x\xi).
\end{equation}

For $K_\nu$, the modified Bessel function of the second kind already used
in the radial measure, the Mellin integral
\begin{equation}
 \int_0^\infty
 t^{\mu-1}K_\nu(t)\dd t
 =
 2^{\mu-2}
 \Gamma\!\left(\frac{\mu-\nu}{2}\right)
 \Gamma\!\left(\frac{\mu+\nu}{2}\right)
\end{equation}
gives the moment identity \eqref{eq:measure_moments}.

Finally, the Laguerre generating identity used in
\eqref{eq:Bkernel} is
\begin{equation}
 \sum_{j=0}^\infty
 \frac{(-\xi)^j}{(\beta)_j}
 L_j^{(\beta-1)}(x)
 =
 \e^{-\xi}{}_0F_1(; \beta;x\xi).
\end{equation}

\section{A useful operator-domain formulation}
\label{app:operator_domain}

Because $G(x)\sim\sqrt{2x}$, the loss operator is unbounded.  The rigorous
definition of
\begin{equation}
 G(\mathcal J_\ell)
\end{equation}
is the spectral functional calculus of the self-adjoint radial-energy operator
$\mathcal J_\ell$, unitarily equivalent to the Jacobi matrix $J_\ell$; it is
not formal substitution of a differential operator into a Taylor series.  The latter is valid on suitable analytic vectors and is useful
for algebraic calculations, but the spectral definition is primary.

The natural domain is
\begin{equation}
 \Dom(G(\mathcal J_\ell))
 =
 \left\{
 f:
 \int_0^\infty
 |G(x)|^2
 |(\mathcal U_\ell f)(x)|^2
 \dd\rho_\ell(x)
 <\infty
 \right\}.
\end{equation}
Since the gain part is bounded and compact in the standard hard-sphere
Hilbert-space realization, the full radial operator has the same domain.

\section*{Acknowledgement of AI assistance}
During the development and preparation of this manuscript, the author used
ChatGPT (OpenAI, GPT-5.6 Sol) as an interactive research and writing assistant.
Its use included discussion and critical examination of mathematical
formulations, cross-checking of derivations, development of the presentation
and organization of the manuscript, identification of notational and
expository inconsistencies, and assistance with editorial revision and
bibliographic verification.  All mathematical results, arguments,
interpretations, references, and final text were reviewed and accepted by the
author, who assumes full responsibility for the content of the manuscript.

\end{document}